\documentclass[11pt, english]{article}
\usepackage[margin=1in]{geometry}

\usepackage{epsf}
\usepackage{fancyhdr}
\usepackage{graphics}
\usepackage{graphicx}
\usepackage{psfrag}

\usepackage[linesnumbered,ruled]{algorithm2e}% http://ctan.org/pkg/algorithm2e
\DontPrintSemicolon

\usepackage{color}

\usepackage{amsthm}
\usepackage{amsfonts}
\usepackage{amsmath}
\usepackage{amssymb,bbm}
\usepackage{mathrsfs}
\allowdisplaybreaks

\usepackage{dsfont}
\usepackage{euscript}
\usepackage{nicefrac}
\usepackage{pifont}
\usepackage{xspace}
\usepackage{natbib}

\usepackage{tikz-cd}

\usepackage[colorlinks=true,linkcolor=red!50!black,citecolor=blue!50!black,urlcolor=blue!80!black]{hyperref}

\usepackage{chngpage}
\usepackage{tabularx}%

\usepackage{enumitem}
\usepackage{booktabs}
\usepackage{caption,subcaption}

\usepackage{mathtools}
\newcommand{\interior}{\operatorname{int}}
\newcommand{\cE}{\mathcal{E}}
\newcommand{\cT}{\mathcal{T}}
\newcommand{\cJ}{\mathcal{J}}
\newcommand{\cQ}{\mathcal{Q}}

\newcommand{\proposal}{\mathcal{P}}

\newcommand{\step}{\eta}
\newcommand{\tmix}{\tau_{\text{mix}}}

\newcommand{\D}{\mathrm{d}}

\newcommand{\F}{\textup{F}}

\newcommand{\defn}{:=}

\theoremstyle{plain}

\newtheorem{theorem}{Theorem}
\newtheorem{proposition}{Proposition}
\newtheorem{lemma}{Lemma}
\newtheorem{corollary}{Corollary}
\newtheorem{definition}{Definition}

\newtheorem{fact}{Fact}
\newtheorem{remark}{Remark}

\newlength{\widebarargwidth}
\newlength{\widebarargheight}
\newlength{\widebarargdepth}

\makeatletter
\long\def\@makecaption#1#2{
        \vskip 0.8ex
        \setbox\@tempboxa\hbox{\small {\bf #1:} #2}
        \parindent 1.5em  %% How can we use the global value of this???
        \dimen0=\hsize
        \advance\dimen0 by -3em
        \ifdim \wd\@tempboxa >\dimen0
                \hbox to \hsize{
                        \parindent 0em
                        \hfil
                        \parbox{\dimen0}{\def\baselinestretch{0.96}\small
                                {\bf #1.} #2
                                }
                        \hfil}
        \else \hbox to \hsize{\hfil \box\@tempboxa \hfil}
        \fi
        }
\makeatother

\long\def\comment#1{}
\definecolor{battleshipgrey}{rgb}{0.52, 0.52, 0.51}
\definecolor{darkgray}{rgb}{0.66, 0.66, 0.66}
\definecolor{darkgreen}{rgb}{0.0, 0.2, 0.13}
\definecolor{darkspringgreen}{rgb}{0.09, 0.45, 0.27}
\definecolor{dukeblue}{rgb}{0.0, 0.0, 0.61}
\definecolor{olivedrab7}{rgb}{0.24, 0.2, 0.12}
\definecolor{darkblue}{rgb}{0.0, 0.0, 0.55}
\definecolor{darkscarlet}{rgb}{0.34, 0.01, 0.1}
\definecolor{candyapplered}{rgb}{1.0, 0.03, 0.0}
\definecolor{ao(english)}{rgb}{0.0, 0.5, 0.0}
\definecolor{applegreen}{rgb}{0.55, 0.71, 0.0}

\makeatletter
\def\moverlay{\mathpalette\mov@rlay}
\def\mov@rlay#1#2{\leavevmode\vtop{%
   \baselineskip\z@skip \lineskiplimit-\maxdimen
   \ialign{\hfil$\m@th#1##$\hfil\cr#2\crcr}}}
\newcommand{\charfusion}[3][\mathord]{
    #1{\ifx#1\mathop\vphantom{#2}\fi
        \mathpalette\mov@rlay{#2\cr#3}
      }
    \ifx#1\mathop\expandafter\displaylimits\fi}
\makeatother

\newcommand{\cupdot}{\charfusion[\mathbin]{\cup}{\cdot}}

\DeclareMathOperator{\diag}{diag}
\DeclareMathOperator{\Diag}{Diag}

\DeclareMathOperator{\rank}{rank}

\newcommand{\naturalnum}{\ensuremath{\mathbb{N}}}

\newcommand{\Ind}{\ensuremath{\mathbb{I}}}
\newcommand{\Prob}{\ensuremath{{\mathbb{P}}}}

\newcommand{\Normal}{\ensuremath{\mathsf{N}}}

\DeclarePairedDelimiterX{\infdivx}[2]{(}{)}{%
  #1\;\delimsize\|\;#2%
}

\newcommand{\parenth}[1]{\left( #1 \right)}

\newcommand{\abs}[1]{| #1 |}

\newcommand{\tp}{^\mathsf{T}}

\newcommand{\vecnorm}[2]{\left\| #1\right\|_{#2}}

\def\balign#1\ealign{\begin{align}#1\end{align}}
\def\baligns#1\ealigns{\begin{align*}#1\end{align*}}
\def\balignat#1\ealign{\begin{alignat}#1\end{alignat}}
\def\balignats#1\ealigns{\begin{alignat*}#1\end{alignat*}}
\def\bitemize#1\eitemize{\begin{itemize}#1\end{itemize}}
\def\benumerate#1\eenumerate{\begin{enumerate}#1\end{enumerate}}

\newenvironment{talign*}
 {\let\displaystyle\textstyle\csname align*\endcsname}
 {\endalign}
\newenvironment{talign}
 {\let\displaystyle\textstyle\csname align\endcsname}
 {\endalign}

\def\balignst#1\ealignst{\begin{talign*}#1\end{talign*}}
\def\balignt#1\ealignt{\begin{talign}#1\end{talign}}
\def\tinycitep*#1{{\tiny\citep*{#1}}}
\def\tinycitealt*#1{{\tiny\citealt*{#1}}}
\def\tinycite*#1{{\tiny\cite*{#1}}}
\def\smallcitep*#1{{\scriptsize\citep*{#1}}}
\def\smallcitealt*#1{{\scriptsize\citealt*{#1}}}
\def\smallcite*#1{{\scriptsize\cite*{#1}}}

\makeatletter
\@ifundefined{color@orange}{\definecolor{orange}{rgb}{1,0.5,0}}{}
\@ifundefined{color@purple}{\definecolor{purple}{rgb}{0.5,0,0.5}}{}
\@ifundefined{color@teal}{\definecolor{teal}{rgb}{0,0.5,0.5}}{}
\@ifundefined{color@Maroon}{\definecolor{Maroon}{rgb}{0.5,0,0}}{}
\makeatother

\makeatletter
\@ifundefined{color@OliveGreen}{\definecolor{OliveGreen}{rgb}{0,0.6,0}}{}
\makeatother

\global\long\def\on#1{\operatorname{#1}}%

\global\long\def\dw{\mathsf{Dikin\ walk}}%

\global\long\def\Otilde{\widetilde{O}}%

\global\long\def\E{\mathbb{E}}%

\global\long\def\R{\mathbb{R}}%
\global\long\def\Rd{\mathbb{\R}^{d}}%

\global\long\def\veps{\varepsilon}%
\global\long\def\K{\mathcal{K}}%

\global\long\def\op{\textup{op}}%

\global\long\def\poly{\on{poly}}%

\global\long\def\inter{\on{int}}%

\global\long\def\cols{\on{col}}%

\global\long\def\T{\mathsf{T}}%

\global\long\def\mc#1{\mathcal{#1}}%

\global\long\def\msf#1{\mathsf{#1}}%

\global\long\def\Dd{\mathrm{D}}%

\global\long\def\bpar#1{\bigl(#1\bigr)}%
\global\long\def\Bpar#1{\Bigl(#1\Bigr)}%

\global\long\def\babs#1{\bigl|#1\bigr|}%
\global\long\def\bnorm#1{\bigl\Vert#1\bigr\Vert}%
\global\long\def\bbrack#1{\bigl[#1\bigr]}%

\global\long\def\sbrace#1{\{#1\}}%
\global\long\def\bbrace#1{\bigl\{#1\bigr\}}%

\global\long\def\inner#1{\langle#1\rangle}%

\global\long\def\norm#1{\lVert#1\rVert}%

\global\long\def\mmid{\mathbin{\|}}%

\global\long\def\dtv{d_{\msf{TV}}}%
\global\long\def\tv{\msf{TV}}%
\global\long\def\Tr{\on{Tr}}%
\global\long\def\intk{\inter\K}%

\makeatletter
\renewcommand{\paragraph}{%
  \@startsection{paragraph}{4}%
  {\z@}{1.25ex \@plus 1ex \@minus .2ex}{-1em}%
  {\normalfont\normalsize\bfseries}%
}
\makeatother

\newenvironment{acknowledgements}
{\paragraph{Acknowledgements.}}
{}

\title{On two proofs of $d^2$ mixing of weighted Dikin walks\author{Yuansi Chen\\ ETH Z\"urich\\ \texttt{yuansi.chen@stat.math.ethz.ch} \and Yunbum Kook\\ University of Michigan\\  \texttt{ybkook@umich.edu}}}
\date{  }

\begin{document}

\maketitle

\begin{abstract}
We study the mixing time of weighted Dikin walks for sampling from exponential distributions on polytopes and truncated positive-semidefinite (PSD) cones.
Our first result gives a general total-variation mixing bound under strong self-concordance, $\bar{\nu}$-symmetry, and mixed-trace regularity on the local metric.
The key idea is to control the Metropolis--Hastings acceptance probability on a high-probability region rather than at every point.
Applying this framework to the Lee--Sidford, Lewis-weight, and John metrics yields an $\widetilde O(d^2)$ mixing bound for sampling from polytopes, while applying it to a hybrid barrier yields an $\widetilde O(d^4)$ mixing bound for sampling from truncated PSD cones.
Our second result establishes stronger $\chi^2$-divergence guarantees and pointwise acceptance control using a new fourth-order bootstrap condition.
For a suitably scaled Lee--Sidford metric, this yields an $\widetilde O(d^2)$ mixing bound in $\chi^2$-divergence, improving on the previous $\widetilde O(d^{9/4})$ bound.
\end{abstract}

\section{Introduction}

Sampling uniformly from a high-dimensional convex body is a fundamental primitive in polynomial-time randomized algorithms for computing the volume of the body~\cite{DFK91random,LS93random,KLS97random,LV06simulated,CV18Gaussian,JLLV26reducing}.
A long line of work has studied the complexity of this algorithmic task assuming access to a membership oracle for the convex body $\K$, which answers whether a queried point lies in $\K$.
This abstract computational model has provided a common setting for designing algorithms and analyzing their complexity.
For instance, general-purpose samplers such as the Ball walk~\cite{Lovasz90compute,LS93random,KLS97random}, Hit-and-Run~\cite{Lovasz99hit,LV06hit,KV26HAR}, and In-and-Out~\cite{KVZ26INO,KV25sampling,KV26zeroLC,KV26unified} have been developed and analyzed.

In practice, we often know an \emph{explicit description} of a convex body.
For instance, a polytope---finitely many linear inequalities---is a canonical convex body, and uniform sampling from a polytope has found applications in Bayesian statistics, machine learning, and systems biology.
It is natural to ask whether one can leverage this additional information to design more efficient sampling algorithms.

Kannan and Narayanan~\cite{KN12random} proposed the Dikin walk, which exploits the inequality description to construct position-dependent Gaussian proposals.
Inspired by interior-point methods (IPMs) for structured convex optimization~\cite{Dikin1967iterative,Karmarkar84polytime,NN94ipm}, the walk uses the Hessian of a barrier function as its local metric.
For a bounded, full-dimensional polytope $\K \defn \{x\in\Rd:Ax\geq b\}$ with $A\in\R^{m\times d}$ and $b\in\R^m$, the Hessian of the logarithmic barrier $\phi_{\log}(x)\defn-\sum_{i=1}^m\log(a_i\tp x-b_i)$ encodes the local geometry of the polytope, and Kannan and Narayanan showed that the Dikin walk with this local geometry mixes in $\Otilde(md)$ iterations from an $O(1)$-warm start~\cite{KN12random,SV16mixing}.

This $m$-dependence is \emph{representation-dependent}, since redundant constraints can increase $m$ without changing $\K$. 
The optimization literature has addressed this issue through Vaidya's volumetric-logarithmic barrier and the Lewis-weight-based Lee--Sidford barrier~\cite{Vaidya96convex,CP15lprow,LS19solving}, which improve the dependence of IPM iteration complexity on $m$.
This suggests an analogous route to reduce the dependence on $m$ in sampling with the Dikin walk.

Implementing these ideas, Chen, Dwivedi, Wainwright, and Yu introduced the Vaidya and John walks~\cite{CDW18MCMC}.  The Vaidya walk, based on the volumetric-logarithmic barrier, has an $\Otilde(m^{1/2}d^{3/2})$ warm-start mixing bound.  The John walk, based on approximate John ellipsoids, has an $\Otilde(d^{5/2})$ warm-start mixing bound.
They further conjectured that a refined analysis would improve the John-walk bound to $\Otilde(d^2)$.

An $\Otilde(d^2)$ mixing bound is consistent with progress in both optimization and general-purpose sampling.
For linear programming, the Lee--Sidford barrier improves the interior-point iteration bound from $\Otilde(\sqrt m)$ for the logarithmic barrier to $\Otilde(\sqrt d)$, leaving only polylogarithmic dependence on the number of constraints~\cite{LS19solving}. 
For uniform sampling, general-purpose samplers achieve $\Otilde(d^2)$ mixing from an $O(1)$-warm start when the body is isotropic (i.e., the covariance is close to the identity)~\cite{KLS97random,Klartag23log,KVZ26INO,KV26HAR}.
Since the Dikin walk is affine-invariant, one may think of it as enforcing `local' isotropy, making $\Otilde(d^2)$ a natural goal for a weighted Dikin walk.

Laddha, Lee, and Vempala proposed a route to this goal based on \emph{strong self-concordance} and \emph{$\bar\nu$-symmetry}~\cite{LLV20strong}.  
Strong self-concordance is a stronger version of usual self-concordance required for an improved analysis, while $\bar\nu$-symmetry measures how well the local Dikin ellipsoid approximates the largest centrally symmetric subset of the body centered at the current point. 
Their conductance argument would give an $\Otilde(d\bar\nu)$ warm-start mixing bound if transition overlap held for nearby points.
Such transition overlap, together with their estimate $\bar\nu=\Otilde(d)$ for the Lee--Sidford metric, would imply $\Otilde(d^2)$ mixing. 
Their Metropolis--Hastings analysis, however, does not establish this transition overlap, leaving the implication incomplete\footnote{Detailed expositions can be found in \cite{Kook26Dikin}}.
Later works~\cite{KV24ipm,JC25regularized} analyzed the Dikin walk with the Lewis-weight metric, but their mixing bounds remain at $\Otilde(d^{2.5})$.
Recently, using the Lee--Sidford metric, \cite{Kook26Dikin} showed that a higher-order analysis yields an $\Otilde(d^{9/4})$ warm-start mixing bound for exponential sampling, making progress toward the conjectured $d^2$ bound.

We resolve this conjecture with two main results. 
First, Theorem~\ref{theorem:generic-weighted-Dikin-mixing} proves an $\Otilde(d\bar\nu)$ total-variation mixing bound for exponential sampling under fewer assumptions than prior work~\cite{KV24ipm,JC25regularized}.
Unlike prior work, we show that the Dikin proposal has high acceptance probability for \emph{most} points rather than at every point; this allows us to prove the mixing bound under weaker assumptions.
This framework allows us to easily verify the conditions for the Lee--Sidford, Lewis-weight, and John metrics, yielding an $\Otilde(d^2)$ mixing bound.
Moreover, it gives an $\Otilde(n^2)$ mixing bound for a hybrid metric on truncated positive-semidefinite cones of ambient dimension $n=\Theta(d^2)$, improving on previous mixing bounds~\cite{Narayanan16randomized,KV24ipm}.
Second, Theorem~\ref{theorem:generic-pointwise-mixing} imposes additional assumptions to obtain high acceptance probability at every point, resulting in a stronger $\chi^2$-divergence guarantee.
This strengthens the $\Otilde(d^2)$ mixing bound for the Lee--Sidford Dikin walk from total variation to $\chi^2$-divergence.

Table~\ref{tab:framework-comparison} summarizes the two frameworks and their applications.
\begin{table}[t]
    \centering
    \caption{Comparison of the two generic mixing frameworks.}
    \label{tab:framework-comparison}
    \small
    \setlength{\tabcolsep}{3pt}
    \renewcommand{\arraystretch}{1.2}
    \begin{tabularx}{\textwidth}{@{}l
        >{\hsize=0.9\hsize\raggedright\arraybackslash}X
        >{\hsize=1.5\hsize\raggedright\arraybackslash}X
        >{\hsize=0.8\hsize\raggedright\arraybackslash}X
        >{\hsize=0.6\hsize\raggedright\arraybackslash}X
        >{\hsize=1.2\hsize\raggedright\arraybackslash}X@{}}
        \toprule
        Framework & Acceptance control & Metric assumptions & Conductance & Guarantee & Applications \\
        \midrule
        Theorem~\ref{theorem:generic-weighted-Dikin-mixing}
            & High-prob. region
            & SSC, $\bar\nu$-symmetry, mixed $23$-trace
            & $s$-conductance
            & TV
            & LS, Lewis-weight, John, hybrid PSD \\
        Theorem~\ref{theorem:generic-pointwise-mixing}
            & Pointwise
            & Above conditions, LTSC, fourth-order bootstrap
            & Conductance
            & $\chi^2$
            & Scaled LS \\
        \bottomrule
    \end{tabularx}
\end{table}

\section{Preliminaries}

This section introduces the basic setup, the Dikin walk, regularity conditions, and the weighted metrics used in our applications.

\subsection{Basic setup and notation}

$a\lesssim b$ means $a\leq Cb$ for a universal constant $C>0$, and similarly for $\gtrsim$. We use $a\asymp b$ if $a\lesssim b$ and $b\lesssim a$. 
We use $\Otilde(\cdot)$ to suppress factors polylogarithmic in $m$, $d$, $\bar\nu$, the inverse accuracy, and the warm-start parameter.  Here $\bar\nu$ is the symmetry parameter introduced in Definition~\ref{def:nu-symmetry}. Other standard notation is clarified in~\S\ref{app:deferred_preliminaries}.

\paragraph{Convex bodies and polytopes.}
A convex body is a compact convex set with nonempty interior.  Unless stated otherwise, $\K\subset\Rd$ is a full-dimensional convex body, and $\interior\K$ and $\mathcal B(\K)$ denote its interior and Borel $\sigma$-algebra, respectively.
We denote a polytope by $\K = \{x \in \Rd: Ax \geq b\}$, where $A$ has full column rank and no zero rows, and $b\in \R^m$.

\paragraph{Probability.}
We use the same symbol for a probability measure and its Lebesgue density when clear from context.
For probability measures $\mu$ and $\nu$, their total-variation (TV) distance is $\dtv(\mu,\nu) \defn \sup_E\abs{\mu(E)-\nu(E)}$.
Their chi-square divergence is $\chi^2(\mu\mmid\nu) \defn \int (\frac{\D\mu}{\D\nu}-1)^2\,\D\nu$ if $\mu\ll\nu$; otherwise $\chi^2(\mu\mmid\nu)=\infty$. Note that $4\dtv^2\leq \chi^2$.

Throughout the paper, we aim to sample from an exponential distribution on $\K$: for $\theta\in\Rd$,
\[
    \pi_\theta(\D x)
    \defn
    \frac{1}{Z_\theta}\exp(-\theta\tp x)\,\Ind_\K(x)\,\D x
    \quad\text{for}\quad 
    Z_\theta
    \defn
    \int_\K\exp(-\theta\tp x)\,\D x\,,
\]
where $\Ind_\K$ is the indicator function of $\K$. Basic facts about Markov chains and mixing are deferred to \S\ref{app:deferred_preliminaries}.

\subsection{Local metrics and the Dikin walk}

\paragraph{Local metric.}
Let $g:\interior\K\to\mathbb S_{++}^d$ be a local metric.  We identify symmetric matrices with their bilinear forms and write $g(x)[u,v]\defn u\tp g(x)v$. 
For $h\in\Rd$, its local norm at $x$ is $\norm{h}_{g(x)}^2 \defn h\tp g(x)h$.
The corresponding \emph{Dikin ellipsoid} of radius $\rho>0$ is
\[
    \msf D_g(x,\rho)
    \defn
    \{y\in\Rd:\norm{y-x}_{g(x)}\leq\rho\}\,.
\]
For the Hessian metric of a self-concordant barrier, the unit Dikin ellipsoid is contained in $\K$ \cite[Theorem~2.1.1]{NN94ipm}.

\paragraph{Dikin walk.}
Given a local metric $g$ and a proposal radius $r>0$, set $\step\defn r/\sqrt d$.  From a current state $x\in\interior\K$, one step of the Dikin walk proceeds as follows.
\begin{enumerate}
    \item Sample a proposal
        $Y\sim\proposal_x\defn\Normal(x,\step^2g(x)^{-1})$, whose density is
        \begin{equation}\label{eq:proposal-density}
            p_x(y)
            =\frac{\sqrt{\det g(x)}}{(2\pi\step^2)^{d/2}}
            \exp\bpar{-\frac{1}{2\step^2}\,\norm{y-x}_{g(x)}^2}\,.
        \end{equation}
    \item If $Y\in\interior\K$, move to $Y$ with the Metropolis--Hastings acceptance probability $\alpha_g(x,Y)$ defined below; otherwise remain at $x$.
\end{enumerate}
Equivalently, the proposal can be written as $Y=x+\step H$ for $H\mid X=x\sim \msf N(0,g(x)^{-1})$.

For $x,y\in\interior\K$, we call $p_y(x)/p_x(y)$ the proposal-density ratio and its logarithm the log proposal-density ratio. For the target $\pi_\theta$, the acceptance probability can be written as 
\begin{equation}
\begin{aligned}
    \alpha_g(x,y)
    & =
    \begin{cases}
    \displaystyle
    \min\bbrace{1,
    \exp(-\inner{\theta,y-x})\, \frac{p_y(x)}{p_x(y)}}\,,
    & y\in\interior\K\,,\\[4pt]
    0\,, & y\notin\interior\K\,,
    \end{cases}\\
    \frac{p_y(x)}{p_x(y)}
    & =
    \bpar{\frac{\det g(y)}{\det g(x)}}^{1/2}
    \exp\bbrace{
    -\frac{1}{2\step^2}\,(y-x)\tp
    \bpar{g(y)-g(x)}(y-x)
    }\,.
\end{aligned}
\label{eq:exponential-acceptance}
\end{equation}

Let $\cT$ be the resulting transition kernel, and set $\cT_x\defn\cT(x,\cdot)$. The Metropolis correction makes $\cT$ reversible with respect to $\pi_\theta$, so $\pi_\theta$ is stationary. We use $\widetilde{\cT}$ for its lazy version.

\paragraph{Dikin metrics on polytopes.}
For a polytope $\K = \{Ax\geq b\}$ and $x\in\interior\K$, define
\[
    s_x \defn Ax-b\in\R_{>0}^m\,,
    \qquad
    S_x \defn \Diag s_x\,,
    \qquad
    A_x\defn S_x^{-1}A\,.
\]
Writing $a_i$ and $a_{x,i}$ for the $i$-th rows of $A$ and $A_x$, respectively, we have $a_{x,i}=a_i/s_{x,i}$.  The \emph{weighted Dikin metrics} considered in this paper have the form
\[
    g(x)
    \defn
    A_x\tp W_xA_x\,,
    \qquad
    W_x\defn\Diag w_x\,,
    \qquad
    w_x\in\R_{>0}^m\,.
\]
The Dikin walk with this metric is called a \emph{weighted Dikin walk}.

\subsection{Regularity conditions for mixing}\label{subsec:regularity-conditions}

We first state the regularity conditions used in the high-probability conductance framework and then recall two additional conditions used in the pointwise acceptance analysis.

\paragraph{Self-concordance and symmetry.}
\emph{Self-concordance} is standard in interior-point theory~\cite{NN94ipm} and Dikin-walk analysis~\cite{KN12random}.
\emph{Strong self-concordance} and \emph{$\bar\nu$-symmetry} were introduced for Dikin-walk sampling by Laddha, Lee, and Vempala~\cite{LLV20strong}.
Below, for a differentiable map $F$, $\Dd F(x)[h]$ denotes the first directional derivative of $F$ at $x$ in the direction $h\in\Rd$.

\begin{definition}[Self-concordance and strong self-concordance]\label{def:strongly-self-concordant}
Let $g:\interior\K\to\mathbb S_{++}^d$ be continuously differentiable $(C^1)$. The metric $g$ is \emph{self-concordant} if, for every $x\in\interior\K$ and $h\in\Rd$,
\[
    \norm{
        g(x)^{-1/2}\,\Dd g(x)[h]\,g(x)^{-1/2}
    }_{\op}
    \leq
    2\,\norm{h}_{g(x)}\,.
\]
Equivalently, $-2g(x)\norm{h}_{g(x)} \preceq \Dd g(x)[h] \preceq 2g(x)\norm{h}_{g(x)}$.

It is called \emph{strongly self-concordant} (SSC) if the stronger Frobenius-norm bound holds:
\[
    \norm{
        g(x)^{-1/2}\,\Dd g(x)[h]\,g(x)^{-1/2}
    }_{\F}
    \leq
    2\,\norm{h}_{g(x)}\qquad\text{for any }x\in\interior\K \text{ and } h\in\Rd\,.
\]
\end{definition}

\begin{definition}[$\bar\nu$-symmetry]\label{def:nu-symmetry}
For $\bar\nu\geq1$, a local metric $g$ on a convex body $\K$ is \emph{$\bar\nu$-symmetric} if
\[
    \msf D_g(x,1)
    \subseteq
    \K\cap(2x-\K)
    \subseteq
    \msf D_g(x,\sqrt{\bar\nu})\qquad \text{for any }x\in\interior\K\,,
\]
where $2x-\K$ is the reflection of $\K$ about $x$.
\end{definition}

\paragraph{Mixed-trace regularity.}
For a $C^1$ local metric $g$ and $x\in\interior\K$, define the normalized derivative tensor $\cJ_x^g:(\Rd)^3\to\R$ by
\[
    \cJ_x^g[u,v,w]
    \defn
    (g(x)^{-1/2}u)\tp
    \Dd g(x)[g(x)^{-1/2}w]\,
    g(x)^{-1/2}v\,.
\]
This tensor is symmetric in its first two arguments. For an orthonormal basis $\{e_i\}_{i=1}^d$, define
\begin{align*}
    \vecnorm{\cJ_x^g}{\sbrace{12}\sbrace{3}}
    &\defn
    \sup_{\norm{w}=1}
    \Bpar{
        \sum_{i,j=1}^d
        \cJ_x^g[e_i,e_j,w]^2
    }^{1/2}
    =
    \sup_{\norm{w}=1}
    \norm{
        g(x)^{-1/2}\,
        \Dd g(x)[g(x)^{-1/2}w]\,
        g(x)^{-1/2}
    }_{\F}\,,\\
    &(\Tr_{12}\cJ_x^g)_k
    \defn
    \sum_{i=1}^d\cJ_x^g[e_i,e_i,e_k]\,,\qquad
    (\Tr_{23}\cJ_x^g)_i
    \defn
    \sum_{j=1}^d\cJ_x^g[e_i,e_j,e_j]\,.
\end{align*}
The norm and the resulting contraction vectors are independent of the orthonormal basis.  When $g$ is clear from context, we write $\cJ_x\defn\cJ_x^g$.  Note that SSC is equivalent to $\vecnorm{\cJ_x^g}{\sbrace{12}\sbrace{3}}\leq2$.

\begin{definition}[Mixed-$23$-trace]\label{def:mixed-trace-regularity}
For $\gamma\geq0$, a $C^1$ local metric $g$ has \emph{$\gamma$-bounded mixed $23$-trace} if
\[
    \norm{\Tr_{23}\cJ_x^g}
    \leq
    \gamma\sqrt d
    \qquad\text{for every }x\in\interior\K\,.
\]
When $\gamma$ is a universal constant, we simply say that $g$ has \emph{bounded mixed $23$-trace}.
\end{definition}

This term naturally arises when we analyze the first-order derivative of some quantity in \eqref{eq:dotTheta-chaos-decomposition}.

\begin{remark}[Mixed trace for Hessian metrics]\label{rmk:regularity-ssc}
SSC gives
\[
    \norm{\Tr_{12}\cJ_x^g}
    \leq
    \sqrt d\,\vecnorm{\cJ_x^g}{\sbrace{12}\sbrace{3}}
    \leq
    2\sqrt d\,.
\]
If $g=\nabla^2 f$ for a $C^3$ function $f$, then $\cJ_x^g$ is fully symmetric. Consequently, $\Tr_{23}\cJ_x^g=\Tr_{12}\cJ_x^g$, so an SSC Hessian metric has $2$-bounded mixed $23$-trace.
\end{remark}

\paragraph{Second-order derivative.}
For a twice-differentiable map $F$, we use $\Dd^2F(x)[h_1,h_2]$ to denote its second directional derivative in the directions $h_1$ and $h_2$.
For a $C^2$ local metric $g$ and $x\in\interior\K$, define the normalized second-derivative tensor
\begin{equation}
    \cQ_x^g[u,v,w,z]
    \defn
    \Dd^2g(x)[g(x)^{-1/2}w,g(x)^{-1/2}z]
    [g(x)^{-1/2}u,g(x)^{-1/2}v]\,.
\label{eq:normalized-second-derivative-tensor}
\end{equation}
It is symmetric in $(u,v)$ and in $(w,z)$.  If $g=\nabla^2f$ for a $C^4$ function $f$, then $\cQ_x^g$ is fully symmetric. 
For an orthonormal basis $\{e_i\}_{i=1}^d$, define
\[
    (\Tr_{12}\cQ_x^g)[w,z]
    \defn
    \sum_{i=1}^d\cQ_x^g[e_i,e_i,w,z]\,.
\]
For $h=g(x)^{-1/2}w$,
\begin{align}
    (\Tr_{12}\cQ_x^g)[w,w]
    &=
    \Tr\parenth{g(x)^{-1}\Dd^2g(x)[h,h]},
    \label{eq:Q-trace-second-derivative}\\
    \Dd^2\log\det g(x)[h,h]
    &=
    (\Tr_{12}\cQ_x^g)[w,w]
    -
    \norm{\cJ_x^g[\,\cdot,\cdot,w]}_{\F}^2\,.
    \label{eq:logdet-JQ}
\end{align}

Kook and Vempala~\cite{KV24ipm} introduced lower-trace self-concordance and the original one-sided form of average self-concordance. We use the two-sided form stated by Jiang and Chen~\cite{JC25regularized}.

\begin{definition}[Lower-trace self-concordance]\label{def:ltsc}
Let $g:\interior\K\to\mathbb S_{++}^d$ be $C^2$.  It is \emph{lower-trace self-concordant} (LTSC) if for every $x\in\interior\K$,
\[
    \Tr_{12}\cQ_x^g
    \succeq
    -I_d\,.
\]
Equivalently, $\Tr(g(x)^{-1}\Dd^2g(x)[h,h]) \geq -\norm{h}_{g(x)}^2$ for every $x \in \intk$ and $h\in\Rd$.
\end{definition}

\begin{definition}[Average self-concordance]\label{def:asc}
A local metric $g$ is \emph{average self-concordant} (ASC) if, for every
$\veps\in(0,1/2)$, there is a radius $r_\veps>0$, depending only on
$\veps$ and not on $d$, such that for every $x\in\interior\K$ and
$0<r\leq r_\veps$,
\[
    \Prob_{Y\sim\Normal(x,\frac{r^2}d\,g(x)^{-1})}
    \bpar{
        Y\in\interior\K
        \ \text{ and }\
        \abs{
            \norm{Y-x}_{g(Y)}^2-\norm{Y-x}_{g(x)}^2
        }
        \leq
        \frac{2\veps r^2}{d}
    }
    \geq
    1-\veps\,.
\]
\end{definition}

\subsection{Weighted Dikin metrics}\label{subsec:special-weighted-Dikin-metrics}

We now define the special metrics used in the polytope applications.

\paragraph{Leverage scores.}
For a full-column-rank matrix $B\in\R^{m\times d}$, let $\sigma(B) \defn \diag(B(B\tp B)^{-1}B\tp) \in[0,1]^m$ denote its leverage-score vector.  For the weighted Dikin metric, set $B_x\defn W_x^{1/2}A_x$.  All leverage scores of $B_x$ are positive since $A$ has no zero rows.  The identity
\begin{equation}\label{eq:weighted-Dikin-leverage-identity}
    \sigma_i(B_x)
    =
    w_{x,i}\,a_{x,i}\tp g(x)^{-1}a_{x,i}
    \qquad \text{for each } i\in[m]
\end{equation}
shows that
\[
    \msf D_g(x,1)\subseteq\K
    \quad\Longleftrightarrow\quad
    a_{x,i}\tp g(x)^{-1}a_{x,i}\leq1
    \ \text{for every }i
    \quad\Longleftrightarrow\quad
    \sigma(B_x)\leq w_x \quad \text{entrywise}\,.
\]

\paragraph{Lewis weights.}
For $p\geq2$, the \emph{$\ell_p$-Lewis weights} of $A_x$
are the unique vector $w_x^{(p)}\in\R_{>0}^m$ satisfying
\[
    w_x^{(p)}
    =
    \sigma([W_x^{(p)}]^{1/2-1/p}A_x)
    \quad \text{for }W_x^{(p)}\defn\Diag w_x^{(p)}\,.
\]
The fixed-point identity gives $\sum_{i=1}^m w_{x,i}^{(p)}=d$ and $W_x^{(p)}\preceq I_m$, and the weights vary smoothly with $x\in\interior\K$ \cite[Lemmas~22, 24, and~47]{LS19solving}.

\paragraph{(1) Lee--Sidford metric.}
For $p_{\mathrm{LS}} = 2\,(1+\log m)$ and $\Lambda_{\mathrm{LS}} = \frac{p_{\mathrm{LS}}}{2}$, the normalized \emph{Lee--Sidford (LS) metric}~\cite{LS19solving,LLV20strong} is
\begin{equation}\label{eq:LS-metric}
    g_{\mathrm{LS}}(x)
    \defn
    \Lambda_{\mathrm{LS}}^2
    A_x\tp
    [W_x^{(p_{\mathrm{LS}})}]^{1-2/p_{\mathrm{LS}}}
    A_x\,.
\end{equation}

\paragraph{(2) Lewis-weight metric.}
For  $p_{\mathrm{LW}} = 4\vee\log m$ and $\Lambda_{\mathrm{LW}} = (1+\frac{p_{\mathrm{LW}}}{2})\,\sqrt2\,m^{1/(p_{\mathrm{LW}}+2)}$,
the normalized Lewis-weight metric~\cite{KV24ipm} is defined as 
\begin{equation}\label{eq:LW-metric}
    g_{\mathrm{LW}}(x)
    \defn
    \Lambda_{\mathrm{LW}}^2
    A_x\tp W_x^{(p_{\mathrm{LW}})}A_x\,.
\end{equation}

\paragraph{(3) John metric.}
Following Chen et al.~\cite{CDW18MCMC}, set $\kappa = \log_2\frac{2m}{d}$, $\alpha_{\mathrm{J}} = 1-\frac1\kappa$, and $\beta_{\mathrm{J}} = \frac{d}{2m}$.
The John-weight vector $\zeta_x\in\R_{>0}^m$ is the unique positive
solution of $\zeta_x = \sigma(Z_x^{\alpha_{\mathrm{J}}/2}A_x) + \beta_{\mathrm{J}}\boldsymbol1$ for $Z_x\defn\Diag\zeta_x$,
as follows from the strictly convex formulation of Chen et al.~\cite[Eq.\ (12) and Lemma~3(a)]{LS19solving,CDW18MCMC}. With $\lambda_{\mathrm{J}} = 4\,(1+\kappa)^2$, the normalized John metric is
\begin{equation}
    g_{\mathrm{J}}(x)
    \defn
    \lambda_{\mathrm{J}}A_x\tp Z_xA_x.
    \label{eq:J-metric}
\end{equation}

\section{Main results}

We now present two complementary mixing frameworks.
The first establishes acceptance control only on a high-probability region and uses $s$-conductance to obtain total-variation mixing without assumptions on second derivatives of the metric.
The second adds LTSC and fourth-order bootstrap control to
obtain pointwise acceptance, a spectral gap, and $\chi^2$-contraction.

\subsection{Mixing with acceptance control on a high-probability region}

Our first framework uses $s$-conductance once acceptance is controlled on a high-probability region. Consequently, TV-mixing requires only the three conditions below, which involve at most first derivatives of the metric: SSC, $\bar\nu$-symmetry, and bounded mixed $23$-trace.

\begin{theorem}[TV-mixing of the Dikin walk on a convex body]\label{theorem:generic-weighted-Dikin-mixing}
Let $\K\subset\Rd$ be a full-dimensional convex body, let $\theta\in\Rd$, let $g:\interior\K\to\mathbb S_{++}^d$ be a continuously
differentiable local metric, and let $\pi_\theta$ be the exponential target.
Suppose that, for some $\bar\nu\geq1$ and a universal constant
$\gamma\geq0$,
\begin{enumerate}[label=(\roman*)]
    \item $g$ is SSC
    (Definition~\ref{def:strongly-self-concordant});
    \item $g$ is $\bar\nu$-symmetric
    (Definition~\ref{def:nu-symmetry});
    \item $g$ has $\gamma$-bounded mixed $23$-trace
    (Definition~\ref{def:mixed-trace-regularity}).
\end{enumerate}
For any $\veps\in(0,1/2)$ and initial distribution $\nu_0$, define $L \defn \log\frac{e\,(1+\sqrt{d\bar\nu})(1+\chi^2(\nu_0\mmid\pi_\theta))}{\veps^2}$.
For a proposal radius $r\asymp L^{-3/2}$, the lazy Dikin walk satisfies
\begin{equation*}\label{eq:main-generic-mixing-bound}
    \tmix^{\tv}(\veps;\nu_0,\widetilde{\cT})
    =
    O\bpar{
        d\bar\nu L^3
        \log\frac{1 \vee \chi^2(\nu_0\mmid\pi_\theta)}{\veps^2}
    }
    =
    \Otilde\bpar{
        d\bar\nu
        \log^4\frac{1 \vee \chi^2(\nu_0\mmid\pi_\theta)}{\veps^2}
    }\,.
\end{equation*}
\end{theorem}

See \eqref{eq:mixing-tv} for the definition of $\tmix^{\tv}$.
Our argument builds on the work of Laddha, Lee, and
Vempala~\cite{LLV20strong}, who introduced SSC and $\bar\nu$-symmetry and developed arguments based on proposal overlap and cross-ratio isoperimetry.
Their claimed $d\bar\nu$ bound, however, has a gap. Their equation~(2.2) controls the ellipsoid-volume ratio but does not check whether the current point is included in the Dikin ellipsoid centered at the proposal.
Our new ingredient is condition~(iii) which controls acceptance on a high-probability region under the stationary measure.
When $g$ is a Hessian metric, condition~(i) implies condition~(iii) with $\gamma=2$ (Remark~\ref{rmk:regularity-ssc}). Also, when $g$ is the Hessian metric of a self-concordant barrier, the inner inclusion in condition~(ii) is automatic. 
Condition~(ii) supplies the comparison between cross-ratio distance and local norms used in isoperimetry. Applied to the Lee--Sidford, Lewis-weight, and John metrics, Theorem~\ref{theorem:generic-weighted-Dikin-mixing} yields the conjectured $\Otilde(d^2)$ mixing bound for sampling from uniform and exponential targets on polytopes (Corollary~\ref{cor:LS-mixing}).

\subsection{Mixing with pointwise acceptance control}

Theorem~\ref{theorem:generic-weighted-Dikin-mixing} controls acceptance
only on a high-probability region. The standard pointwise argument instead
uses SSC, LTSC, and ASC to obtain transition overlap for nearby points
as shown by Kook and Vempala~\cite{KV24ipm} and
Jiang and Chen~\cite{JC25regularized}. Together with
$\bar\nu$-symmetry, ordinary conductance then gives a $\chi^2$-contraction factor $\exp\parenth{-\Omega(N/(d\bar\nu))}$ after $N$ steps~\cite{Kook26Dikin}.

For the LS metric, Kook's verification of ASC uses an additional $\Otilde(d^{1/4})$ scaling factor. This factor also increases $\bar\nu$ and yields an $\Otilde(d^{9/4})$ mixing-time bound \cite[Theorem~1.1 and \S1.2]{Kook26Dikin}.
We remove this dimension-dependent factor by combining SSC and bounded mixed $23$-trace with the new fourth-order condition below to derive ASC.

\paragraph{Why a bootstrap condition?}
For $F(t)\defn h\tp g(x+th)h$, Taylor's formula reduces its variation along a proposal path (up to $t\lesssim d^{-1/2}$) to the initial cubic term $F'(0)=\cJ_x^g[\xi^{\otimes3}]$ and the integral of $F''(t)=\cQ_{x_t}^g[\xi_t^{\otimes4}]$. SSC and the bounded mixed $23$-trace condition control $F'(0)$ at scale $\sqrt d$, so it remains to control $F''(t)$ at scale $d$.
A direct uniform bound on $F''(t)$ is too crude and, for the LS metric, loses a factor of $\sqrt d$ \cite{CDW18MCMC,KV24ipm,JC25regularized}.

The higher-order analysis in~\cite{Kook26Dikin} suggests the scaling $\abs{F^{(j)}(0)}\lesssim d^{j/2}$ for $j\geq1$.
More precisely, extending the corresponding base-point through order $k$ would yield a $d^{2+1/(k+1)}$ mixing bound; the paper carries this out for $k=3$.
Heuristically, if the base-point estimates held at every order, then the formal Taylor expansion would give
\[
    1+\abs{F(t)-F(0)}
    \leq
    1+\sum_{j\geq1}\frac{t^j}{j!}\,\abs{F^{(j)}(0)}
    \leq
    e^{Ct\sqrt d}\,.
\]
Observe that this bound has the same form as the bound produced by a Gr\"onwall-type inequality. Hence, it motives finding an auxiliary function whose growth can be controlled in terms of itself.

We therefore introduce $\omega_t$ for suitable bootstrapping. For a fixed $\delta$, conditions~\textup{(B1)}--\textup{(B3)} below should be read as follows: \textup{(B1)} initializes $\omega_0=O_\delta(1)$ on a feasible path of length $O_\delta(d^{-1/2})$; \textup{(B2)} keeps $\omega_t=O_\delta(1)$ along this path by Gr\"onwall's inequality; and \textup{(B3)} converts this control into $\abs{F''(t)}=O_\delta(d)$. Thus, the argument closes through self-consistent pathwise control rather than a worst-case bound on the full fourth-order tensor.

\begin{definition}[Fourth-order bootstrap control]\label{def:fourth-order-bootstrap-control}
A $C^2$ local metric $g$ has \emph{fourth-order bootstrap control} if there is a function $\omega:\interior\K\times\Rd\to[1,\infty)$ with the following property. For every $\delta\in(0,1/2)$, there are constants $r_\delta>0$, $\eta_\delta,b_\delta\geq0$, and $c_\delta\geq1$ that depend only on $\delta$ and are uniform in the dimension and all other problem parameters. For $x\in\interior\K$, let
\[
    \xi\sim\Normal(0,I_d),\qquad
    h=g(x)^{-1/2}\xi,\qquad
    x_t=x+th.
\]
There is an event of probability at least $1-\delta$ on which the following hold:
\begin{enumerate}
    \item[\textup{(B1)}]
    \textup{Feasibility and initial bound.}\quad
    The segment $\{x_t:0\leq t\leq r_\delta/\sqrt d\}$ lies in
    $\interior\K$. On this interval, set
    $\xi_t=g(x_t)^{1/2}h$ and $\omega_t=\omega(x_t,h)$.
    Then $t\mapsto\omega_t$ is continuous and $\omega_0\leq c_\delta$;
    \item[\textup{(B2)}]
    \textup{Bootstrap control.}\quad
    For $0\leq t\leq r_\delta/\sqrt d$,
    \[
        \omega_t
        \leq
        \omega_0
        +
        \eta_\delta\sqrt d\int_0^t\omega_s\,\D s;
    \]
    \item[\textup{(B3)}]
    \textup{Fourth-order dominance.}\quad
    For $0\leq t\leq r_\delta/\sqrt d$,
    $\abs{\cQ_{x_t}^g[\xi_t^{\otimes4}]}
    \leq b_\delta d\,\omega_t$ (see \eqref{eq:normalized-second-derivative-tensor}).
\end{enumerate}
Here, $r_\delta$ is the path radius, $\eta_\delta$ controls the cumulative
growth in condition~\textup{(B2)}, $b_\delta$ is the dominance coefficient
in condition~\textup{(B3)}, and
$c_\delta$ bounds the initial value $\omega_0$.
\end{definition}

Lemma~\ref{lem:fourth-order-bootstrap-ASC} formalizes this argument and shows that SSC, bounded mixed $23$-trace, and fourth-order bootstrap control together establish ASC.
For a scaled LS metric, our verification of this criterion leads to an $\Otilde(d^2)$ mixing-time bound in $\chi^2$-divergence.

\begin{theorem}[$\chi^2$-contraction of the Dikin walk on a convex body]\label{theorem:generic-pointwise-mixing}
Let $\K\subset\Rd$ be a bounded, full-dimensional convex body, let $\theta\in\Rd$, let $g:\interior\K\to\mathbb S_{++}^d$ be a $C^2$ local metric, and let $\pi_\theta$ be the exponential target.
Suppose that, for some $\bar\nu\geq1$ and a universal constant $\gamma\geq0$,
\begin{enumerate}[label=(\roman*)]
    \item $g$ is SSC
    (Definition~\ref{def:strongly-self-concordant});
    \item $g$ is $\bar\nu$-symmetric
    (Definition~\ref{def:nu-symmetry});
    \item $g$ has $\gamma$-bounded mixed $23$-trace
    (Definition~\ref{def:mixed-trace-regularity});
    \item $g$ is LTSC
    (Definition~\ref{def:ltsc});
    \item $g$ has fourth-order bootstrap control
    (Definition~\ref{def:fourth-order-bootstrap-control}).
\end{enumerate}
There exist universal constants $r\in(0,1/2]$ and $c>0$ such that, for any initial distribution $\nu_0$, the
lazy Dikin walk with proposal radius $r$ satisfies, for every integer $N\geq0$,
\[
    \chi^2(\nu_0\widetilde{\cT}^{\,N}\mmid\pi_\theta)
    \leq
    \exp\bpar{-\frac{cN}{d\bar\nu}}\,
    \chi^2(\nu_0\mmid\pi_\theta)\,.
\]
In particular, for any
$\veps\in(0,1/2)$,
\[
    \tmix^{\chi^2}(\veps;\nu_0,\widetilde{\cT})
    =
    O\bpar{
        d\bar\nu
        \log\frac{1\vee\chi^2(\nu_0\mmid\pi_\theta)}{\veps}
    }\,.
\]
\end{theorem}

\subsection{Sampling from polytopes and truncated positive-semidefinite cones}
\label{subsec:consequences}

We provide four consequences of the two frameworks. Corollary~\ref{cor:generic-weighted-Dikin-mixing-polytope} specializes Theorem~\ref{theorem:generic-weighted-Dikin-mixing} to weighted Dikin metrics on polytopes. Corollary~\ref{cor:LS-mixing} specializes it to the Lee--Sidford, Lewis-weight, and John metrics on polytopes, while Corollary~\ref{cor:dimension-square-SDP} applies Theorem~\ref{theorem:generic-weighted-Dikin-mixing} to a hybrid of the PSD barrier and the Lee--Sidford metric on bounded truncated PSD cones. Both corollaries give TV-mixing bounds for exponential targets.
Corollary~\ref{cor:spectral-LS-mixing} applies Theorem~\ref{theorem:generic-pointwise-mixing} to obtain $\chi^2$-contraction for a Lee--Sidford metric on polytopes. The required metric regularities are verified in \S\ref{sec:applications}.

\paragraph{Total-variation mixing on polytopes.}
The next corollary gives sufficient conditions for a weighted Dikin metric
on a polytope to satisfy the assumptions of
Theorem~\ref{theorem:generic-weighted-Dikin-mixing}.

\begin{corollary}[Dikin walk specialized to polytopes]\label{cor:generic-weighted-Dikin-mixing-polytope}
Let $\K=\{x\in\Rd:Ax\geq b\}$ be a bounded, full-dimensional polytope, where $A\in\R^{m\times d}$ has full column rank and no zero rows, and $b\in\R^m$. Consider a local metric of the form $g(x)=A_x\tp W_xA_x$, where $w_x\in\R_{>0}^m$, $W_x=\Diag w_x$, and $x\mapsto w_x$ is continuously differentiable.
Let $B_x=W_x^{1/2}A_x$, $\Sigma_x=\Diag\sigma(B_x)$, and define $R_x:\Rd\to\R^m$ by
\begin{align}
    \label{eq:relative-row-derivative}
    R_x(z)
    \defn
    \frac12\,\Dd\log w_x[g(x)^{-1/2}z]
    -
    A_xg(x)^{-1/2}z\,.
\end{align}
Suppose that there are $C\in[0,1]$ and $\bar\nu\geq1$ such that, for every $x\in\interior\K$,
\begin{enumerate}[label=(\roman*)]
    \item $\norm{\Sigma_x^{1/2}R_x}_{\op}\leq C$;
    \item $\msf D_g(x,1)\subseteq\K$, equivalently,
    $a_{x,i}\tp g(x)^{-1}a_{x,i}\leq 1$ for every $i\in[m]$;
    \item $\Tr W_x\leq\bar\nu$.
\end{enumerate}
Then, $g$ is SSC, $\bar\nu$-symmetric, and has
$2C$-bounded mixed $23$-trace. In particular, it has bounded mixed
$23$-trace, so Theorem~\ref{theorem:generic-weighted-Dikin-mixing}
applies.
\end{corollary}
Condition~(i) controls SSC and the mixed $23$-trace bound in the following way: since $g(x)=B_x^\T B_x$ and $\Dd B_x[h]=\Diag(R_xz)B_x$ for $h=g(x)^{-1/2}z$, the bound on $\norm{\Sigma_x^{1/2}R_x}_{\op}$ ensures that $g(x)$ varies slowly.
By~\eqref{eq:weighted-Dikin-leverage-identity},
condition~(ii) is equivalent to the pointwise leverage-score bound
$\sigma_i(W_x^{1/2}A_x)\leq w_{x,i}$ for every $i\in[m]$.
Lastly, Conditions~(ii) and~(iii) imply $\bar\nu$-symmetry.

In \S\ref{sec:LS-verification}, we verify these conditions for the Lee--Sidford metric $g_{\mathrm{LS}}$, the Lewis-weight metric $g_{\mathrm{LW}}$, and the John metric $g_{\mathrm{J}}$, and Corollary~\ref{cor:generic-weighted-Dikin-mixing-polytope} implies  the following dimension-square TV-mixing bound.

\begin{corollary}[Dimension-square TV-mixing on polytopes]\label{cor:LS-mixing}
    Let $\K=\{x\in\Rd:Ax\geq b\}$ be a bounded,
    full-dimensional polytope, where $A\in\R^{m\times d}$ has full
    column rank and no zero rows.  Let
    $\pi_\theta\propto\exp(-\theta\tp x)\,\Ind_\K(x)$ be an
    exponential target on $\K$, and let $g$ be any of
    $g_{\mathrm{LS}}$, $g_{\mathrm{LW}}$, or $g_{\mathrm{J}}$, as defined
    above. The corresponding symmetry parameters satisfy
    \begin{align*}
        \bar\nu_g
        \lesssim
        \begin{cases}
            d\log^2 m\,, & g=g_{\mathrm{LS}}\,,\\
            d\log^2m\,, & g=g_{\mathrm{LW}}\,,\\
            d\log^2\frac{em}{d}\,, & g=g_{\mathrm{J}}\,.
        \end{cases}
    \end{align*}
    Given $\veps\in(0,1/2)$ and an initial distribution $\nu_0$, set $L_g\defn\log[d\log(em)(1+\chi^2(\nu_0\mmid\pi_\theta))/\veps^2]$. Then, for some proposal radius, the lazy Dikin walk with local metric $g$ satisfies $\dtv(\nu_0\widetilde{\cT}^{\,N},\pi_\theta)\leq\veps$ after
    \[
        N
        =
        O\bpar{
            d\bar\nu_g L_g^3
            \log\frac{
                1\vee\chi^2(\nu_0\mmid\pi_\theta)
            }{\veps^2}
        }
        \quad\text{steps}\,.
    \]
\end{corollary}

\begin{remark}
For the LS metric, one can instead use a better choice of $p_{\textrm{LS}} = 2 \vee \log \frac{m}{d}$. In this case, the symmetry parameter is improved to $\bar{\nu} \lesssim d \log^2\frac{em}{d}$.
\end{remark}

\paragraph{Total-variation mixing on truncated positive-semidefinite cones.}

Theorem~\ref{theorem:generic-weighted-Dikin-mixing} also applies to
exponential targets on truncated positive-semidefinite cones.
In \S\ref{ssub:dimension_square_mixing_for_psd_cones}, we verify its assumptions for a hybrid metric on the domain
\[
    \K \defn \{X \in \mathbb{S}^d: X\succeq 0\,, \inner{A_i, X}\ge b_i\text{ for }i\in [m] \}\,.
\]
The ambient dimension is $n\defn d(d+1)/2$.
Since the Hessian metric of the positive-semidefinite barrier controls directions outside the span of the linear constraints, the number $m$ of such constraints need not be at least $n$. The corollary below gives an
$\Otilde(d^4)=\Otilde(n^2)$ bound, which is dimension-square in the ambient dimension.

\begin{corollary}[Dimension-square TV-mixing on truncated PSD cones]\label{cor:dimension-square-SDP}
    Let $A_1,\ldots,A_m,\Theta\in\mathbb S^d$ and $b\in\R^m$, and suppose
    that $\K=\{X\in\mathbb S^d:X\succeq0,\ \inner{A_i,X}\geq b_i
    \text{ for }i\in[m]\}$ is bounded and full-dimensional in $\mathbb S^d$.
    Let $\pi_\Theta\propto\exp(-\inner{\Theta,X})\,\Ind_\K(X)$ be the exponential target, $g_{\mathrm{LS}}$ be the LS component defined in \S\ref{ssub:dimension_square_mixing_for_psd_cones}, $\phi_{\mathrm{PSD}}(X)\defn-\log\det X$, and set $g\defn d\nabla^2\phi_{\mathrm{PSD}}+g_{\mathrm{LS}}$.
    Given $\veps\in(0,1/2)$ and an initial distribution $\nu_0$, set $L_g \defn \log\frac{dm\,(1\vee \chi^2(\nu_0\mmid\pi_\Theta))}{\veps^2}$. Then, for a suitable proposal radius, the lazy Dikin walk with local metric $g$ satisfies $\dtv(\nu_0\widetilde{\cT}^{\,N},\pi_\Theta)\leq\veps$ after $N$ steps, where
    \[
        N
        =
        O\bpar{
            d^4  L_g^3\log^2(em)
            \log\frac{
                1\vee\chi^2(\nu_0\mmid\pi_\Theta)
            }{\veps^2}
        }
        \quad\text{steps}\,.
    \]
\end{corollary}

\paragraph{$\chi^2$-contraction on polytopes.}
Finally, we apply Theorem~\ref{theorem:generic-pointwise-mixing} to an LS metric on a polytope.
\S\ref{sec:spectral-LS-verification} verifies fourth-order bootstrap control for this metric. The resulting
$\Otilde(d^2)$ mixing-time bound improves the previous  $\Otilde(d^{9/4})$ bound~\cite{Kook26Dikin}.

\begin{corollary}[$\chi^2$-contraction on polytopes with a scaled LS metric]\label{cor:spectral-LS-mixing}
Let $\K=\{x\in\Rd:Ax\geq b\}$ be bounded and full-dimensional, where $A$ has full column rank and no zero rows. Set $p=2\,(1+\log m)$ and $\lambda=p/2$, let $W_x$ be the $\ell_p$-Lewis-weight matrix of $A_x$, and define $\widehat g(x)\defn S\lambda^2A_x\tp W_x^{1-2/p}A_x$ with $S=C_0(1+\log m)^{20}$ for a sufficiently large universal constant $C_0$.

There exist universal constants $r\in(0,1/2]$ and $c>0$ such that, for every exponential target $\pi_\theta$ and every initial distribution $\nu_0$, the lazy Dikin walk with metric $\widehat g$ and proposal radius $r$  
satisfies
\[
    \chi^2(\nu_0\widetilde{\cT}^{\,N}\mmid\pi_\theta)
    \leq
    \exp\bpar{-\frac{cN}{S\lambda^2d^2}}\,
    \chi^2(\nu_0\mmid\pi_\theta).
\]
In particular, for every $\veps\in(0,1/2)$,
\[
    \tmix^{\chi^2}(\veps;\nu_0,\widetilde{\cT})
    =
    O\bpar{
        d^2\log^{22} (em)
        \log\frac{
            1\vee\chi^2(\nu_0\mmid\pi_\theta)
        }{\veps}
    }\,.
\]
\end{corollary}

\paragraph{AI disclosure.}
The authors used ChatGPT Pro 5.5 and 5.6 for the following purposes:
\begin{itemize}
    \item to obtain a preliminary version of Corollary~2 for Lee--Sidford metric, including a proof of Lemma~4 which is the key to simplify the acceptance-rate control in Theorem~1; and
    \item to develop Corollary~4 based on a ChatGPT 5.6 Pro session shared by Victor Reis (July 16).
\end{itemize}
Both authors then generalized the preliminary Lee--Sidford Dikin-walk proofs into the two generic frameworks of Theorems~1 and~2 and clarified which components of the AI-generated proofs were novel.
ChatGPT Pro was also used to identify errors in the proofs and exposition.

On August~21, Matthew Zhang and Sinho Chewi informed us that Yiping Lu had posted a note on his academic webpage claiming a $d^2$ result for the Lee--Sidford walk by resolving conjectures posed in~\cite{Kook26Dikin}. The result and its proof may overlap with our Corollary~4.
On August~26, we became aware of a paper by Song and Zhang on arXiv that claims a $d^2$ result for the Lee--Sidford walk and may also overlap with  Corollary~4.

\section{Proof of Theorem~\ref{theorem:generic-weighted-Dikin-mixing} via
\texorpdfstring{$s$}{s}-conductance}\label{sec:generic-overlap}

Unlike earlier work on the Dikin walk~\cite{KN12random,CDW18MCMC,LLV20strong,KV24ipm,JC25regularized}, our proof leverages $s$-conductance, which requires (1) one-step coupling and (2) an isoperimetric inequality. In contrast to conductance, $s$-conductance allows us to ignore a (potentially) bad region where the algorithm suffers from slow mixing. However, previous $s$-conductance arguments in the analysis of constrained sampling~\cite{Lovasz99hit,LV07geometry,CE25hitandrun} have resulted in an undesirable \emph{polynomial} dependence on the warmness parameter and inverse accuracy (i.e., $\poly(M/\veps)$).

Previous work has attempted to establish one-step coupling for arbitrary interior points by analyzing the proposal distributions (i.e., $H\sim \mc N(0,g(x)^{-1})$). In the present work, we enlarge the probability space from the proposal variable $H$ to the pair $(X,H)$ and work with a \emph{joint law} $\mu$ defined on this augmented space such that its $X$-marginal and the conditional distribution $H\mid X=x$ correspond to the stationary measure $\pi_\theta$ and the Dikin proposal $\mc P_x$, respectively.

We then define a region $\Xi$ consisting of points from which the acceptance probability of a proposal exceeds a certain threshold. Under a suitable choice of this threshold, we can establish one-step coupling within $\Xi$ under regularity conditions on the local metric. A key technical challenge is to show that $\Xi$ has high probability under $\pi_\theta$ (Lemma~\ref{lem:acceptance_rate_control}). This is where the lifted joint measure $\mu$ comes in. We bound $\pi_\theta(\Xi^c)$ in terms of a ``bad'' event under the joint measure $\mu$, whose probability can in turn be controlled via a moment estimate for the negative log proposal-density ratio under SSC and bounded mixed $23$-trace (Lemmas~\ref{lem:acceptance-to-initial} and~\ref{lem:generic-Rk-bound}). This technique previously appeared in analyses of MALA and HMC by Wu, Schmidler, and Chen~\cite{WSC22minimax} and Chen, Gatmiry, and Jiang~\cite{CG23when}.

The second ingredient, an isoperimetric inequality derived from $\bar \nu$-symmetry, is standard. It follows from self-concordance, the local-norm comparison implied by $\bar\nu$-symmetry, and cross-ratio isoperimetry. Combining the one-step coupling and the isoperimetric inequality, we can bound the $s$-conductance as in Lemma~\ref{lem:generic-warm-start-mixing}, following a standard conductance argument. Finally, using the standard conductance-to-mixing result of Lov\'asz--Simonovits (Proposition~\ref{prop:mixing-from-conductance}), we prove Theorem~\ref{theorem:generic-weighted-Dikin-mixing}.

\subsection{Acceptance control and one-step coupling}\label{sec:closeness_of_transition_kernels}

Recall that $\proposal_x$ and $\cT_x$ are the proposal and one-step transition distributions from $x$. We first use the following standard SSC proposal-overlap estimate \cite{LLV20strong,KV24ipm,JC25regularized}. The proof is standard and recalled in \S\ref{subsec:generic-proposal-overlap}.
\begin{lemma}[Proposal overlap]\label{lem:proposal_overlap}
    Suppose that $g$ is SSC and $0<r\leq 1/2$.  For any
    $x,y\in\interior\K$ such that $\norm{x-y}_{g(x)} \leq \frac{\step}{10} = \frac{r}{10\sqrt d}$, we have $\dtv(\proposal_x, \proposal_y) \leq 1/4$.
\end{lemma}

Next, we control the probability loss from rejections when passing from proposal overlap to transition overlap. Let $(X,H)\sim\mu$ have the joint distribution defined by
\begin{equation}\label{eq:joint-law}
X \sim \pi_\theta\,,\qquad H \mid X = x \sim \Normal\bpar{0, g(x)^{-1}}\,.
\end{equation}
Under $\mu$, define the feasible-proposal and half-space events by
\[
    \cE
    \defn
    \{(x,h):x+\step h\in\interior\K \}\,,\qquad
    \mathcal D_\theta
    \defn
    \{(x,h):\inner{\theta,h}\leq0\}\,.
\]
For $a>0$, define the bad-proposal event by
\begin{equation}\label{eq:bad-event}
    \mathcal B_{\theta,a}
    \defn
        \cE^c
        \cup
        \bbrace{
            (x,h)\in \mathcal D_\theta \cap \cE :
            \frac{p_{x+\step h}(x)}{p_x(x+\step h)}<e^{-a}
        }\,.
\end{equation}
Thus, a bad proposal is either infeasible or feasible and downhill with log-proposal ratio below $-a$.

The next lemma defines a (high-probability) set $\Xi$ on which bad proposals have conditional probability at most $\tau$. It bounds $\pi_\theta(\Xi^c)$ and combines this control with Lemma~\ref{lem:proposal_overlap} to obtain transition overlap between nearby points in $\Xi$. We defer its proof to \S\ref{subsec:generic-acceptance-control}.

\begin{lemma}[Acceptance control and transition overlap]\label{lem:acceptance_rate_control}
    Fix a proposal radius $r>0$ and set $\step=r/\sqrt d$. Let $g$ be SSC and have $\gamma$-bounded mixed $23$-trace on $\interior\K$ for some $\gamma\geq0$.
    Fix $a>0$ and $\tau\in(0,1/8)$. For the bad event $\mathcal B_{\theta,a}$ in \eqref{eq:bad-event}, define
    \[
        \Xi
        \defn
        \bbrace{
            x\in\interior\K:
            \Prob_{H\sim\Normal(0,g(x)^{-1})}
            \bpar{(x,H)\in\mathcal B_{\theta,a}}
            \leq\tau
        }\,.
    \]
    Then, for every $k\geq2$,
    \begin{equation}\label{eq:acceptance-region-complement-mass}
        \pi_\theta(\Xi^c)
        \leq
        \frac{e^a}{\tau}\,
        \bpar{
            \frac{2r(2+\gamma)\,k^{3/2}}{a}
        }^k
        +
        \frac{\mu(\cE^c)}{\tau}\,.
    \end{equation}
    If $r\leq1/2$, then all $x,y\in\Xi$ with $\norm{x-y}_{g(x)}\leq\step/10$ satisfy
    \begin{equation}\label{eq:exponential-transition-overlap}
        \dtv(\cT_x,\cT_y)
        \leq
        1-e^{-a}\,\bpar{\frac14-2\tau}\,.
    \end{equation}
\end{lemma}

For a concrete choice of parameters, let $k\geq2$ and set $a_0=\log\frac76$, $\tau_0=\frac{1}{16}$, and $r = \frac{a_0}{4(2+\gamma)\,k^{3/2}}$ (which ensures $r<1/2$). Substituting these choices into Lemma~\ref{lem:acceptance_rate_control} gives
\begin{equation}
\begin{aligned}
    \pi_\theta(\Xi^c)
    &\leq
    \frac{19}{2^{k}}
    +
    16\mu(\cE^c)\,,\\
    1-\dtv(\cT_x,\cT_y)
    &\geq
    \frac{3}{28}
    \quad\text{whenever }x,y\in\Xi
    \text{ and }
    \norm{x-y}_{g(x)}\leq\frac\step{10}\,.
\end{aligned}
\label{eq:acceptance-region-complement-mass-a0}
\end{equation}

\subsection{From one-step coupling to mixing}\label{subsec:transition-overlap-to-mixing}

We use two standard geometric facts from Dikin-walk conductance arguments \cite{LLV20strong,KV24ipm,JC25regularized}. For distinct $x,y\in\interior\K$, let $p$ and $q$ be the endpoints of the chord through $x$ and $y$, ordered as $p,x,y,q$. The \emph{cross-ratio distance} is
\[
    d_{\K}(x,y)
    \defn
    \frac{
        \norm{x-y}\,\norm{p-q}
    }{
        \norm{p-x}\,\norm{q-y}
    }\,.
\]
For nonempty $A,B\subseteq\K$, define $d_{\K}(A,B)\defn\inf_{x\in A,\,y\in B}d_{\K}(x,y)$.

\begin{fact}[Log-concave cross-ratio isoperimetry, {\cite[Theorem~2.5]{LV07geometry}}]\label{fact:cross-ratio-isoperimetry}
For any log-concave distribution $\pi$ on $\K$ and every measurable partition $\K=S_1\cupdot S_2\cupdot S_3$ with $S_1,S_2\neq\varnothing$,
\[
    \pi(S_3)
    \geq
    d_{\K}(S_1,S_2)\,\pi(S_1)\,\pi(S_2)\,.
\]
\end{fact}

\begin{fact}[Cross-ratio/local-norm comparison, {\cite[Lemma~2.3]{LLV20strong}}]\label{fact:llv-nu-symmetry}
Suppose that a local metric $g$ on $\interior\K$ is $\bar\nu$-symmetric.
Then, for all distinct $x,y\in\interior\K$,
\[
    d_\K(x,y) \geq \frac{1}{\sqrt{\bar\nu}}\,\min\{\norm{x - y}_{g(x)}, \norm{x - y}_{g(y)}\}\,.
\]
\end{fact}

Using this isoperimetric inequality, we turn transition overlap into an $s$-conductance bound.
Recall that for a $\pi_\theta$-reversible Markov kernel $\mathsf T$, its ergodic flow and $s$-conductance are defined by
\begin{align*}
    Q_{\mathsf T}(A,B)
    &\defn
        \int_A\mathsf T_x(B)\,\pi_\theta(\D x)\,,\\
    \Phi_s(\mathsf T)
    &\defn
    \inf_{\pi_\theta(A) \in (s,1/2]}
    \frac{Q_{\mathsf T}(A,A^c)}{\pi_\theta(A)-s}\,,
    \qquad \text{for }0\leq s<\frac12\,.
\end{align*}
Note that $\Phi_0$ is the ordinary conductance. We write $\widetilde{\mathsf T}$ for the lazy version of
$\mathsf T$.

The next standard lemma turns one-step coupling in $\Xi$ into an $s$-conductance bound~\cite{Lovasz99hit,LLV20strong,KV24ipm,JC25regularized}. The proof is deferred to \S\ref{subsec:proof-transition-overlap-conductance}.

\begin{lemma}[Transition overlap to $s$-conductance]\label{lem:generic-warm-start-mixing}
Let $g$ be an SC local metric on $\interior\K$.
Suppose that $\zeta\geq0$, $\Delta\in(0,1]$, $\delta\in[0,1)$, $\bar\nu>0$, and $\Xi\subseteq\interior\K$  satisfy
\begin{align*}
    \pi_\theta(\Xi^c)&\leq\zeta\,,\\
    \dtv(\mathsf T_x,\mathsf T_y)
    &\leq\delta
    \quad\text{whenever }
    x,y\in\Xi
    \ \text{and}\
    \norm{x-y}_{g(x)}\leq\Delta\,,\\
    \norm{x-y}_{g(x)} \wedge \norm{x-y}_{g(y)}
    &\leq
    \sqrt{\bar\nu}\,d_{\K}(x,y) \qquad \text{for distinct }x,y\in\Xi\,.
\end{align*}
Set $\vartheta\defn 1\wedge\frac{\Delta}{2\sqrt{\bar\nu}}$. Then, for every $s \in [\frac{8\zeta}\vartheta, \frac 12)$,
\[
    \Phi_s(\widetilde{\mathsf T}) \geq \frac{(1-\delta)\,\vartheta}{64}\,.
\]
\end{lemma}

We then use conductance to bound the mixing time.
\begin{proposition}[Mixing from conductance
{\cite{lawler1988bounds,LS93random}}]
\label{prop:mixing-from-conductance}
The following holds for any $k \in \mathbb{Z}_{\ge 0}$:
\begin{equation}
\begin{aligned}
    \chi^2(\nu_0\widetilde{\mathsf T}^{\,k}\mmid\pi_\theta)
    &\leq
    \chi^2(\nu_0\mmid\pi_\theta)\,\exp\bpar{-k\Phi_0(\widetilde{\mathsf T})^2}\,,\\
    \dtv(\nu_0\widetilde{\mathsf T}^{\,k}, \pi_\theta)
    &\leq
    \sqrt{s\chi^2(\nu_0\mmid\pi_\theta)}
    +
    \sqrt{\frac{\chi^2(\nu_0\mmid\pi_\theta)}{s}}\,
    \exp\bpar{-\frac{k\Phi_s(\widetilde{\mathsf T})^2}{2}}\,, \qquad 0<s<\frac12\,.
\end{aligned}
\label{eq:warm-start-conductance}
\end{equation}
\end{proposition}

The first estimate is the standard $L^2$ Cheeger bound in the normalization
of Lov\'asz and Simonovits~\cite[Lemma~1.7 and Corollary~1.8]{LS93random}; see also Lawler and Sokal~\cite{lawler1988bounds}.
The second follows from the Lov\'asz--Simonovits $s$-conductance bound
\cite[Corollary~1.6 and Remark~1]{LS93random} by Cauchy--Schwarz, since $\pi_\theta$ is atom-free.

\subsection{Proof of Theorem~\ref{theorem:generic-weighted-Dikin-mixing}}

\begin{proof}[Proof of Theorem~\ref{theorem:generic-weighted-Dikin-mixing}]
    Choose an integer $\ell$ and a radius $r$ satisfying
    \[
        \ell
        \asymp
        \log\frac{
            \sqrt{d\bar\nu}\,\bpar{1\vee \chi^2(\nu_0\mmid\pi_\theta)}
        }{\veps^2}
        \asymp
        L\,,
        \qquad
        r
        \asymp
        \ell^{-3/2}\,,
        \qquad
        \step
        \defn
        \frac r{\sqrt d}\,,
    \]
    where the implicit constants in the choices of $\ell$ and $r$ are taken sufficiently large and sufficiently small, respectively, so that $\ell\geq2$ and $r\leq1/2$.

    The inclusion $\msf D_g(x,1)\subseteq\K$ follows from $\bar\nu$-symmetry. Conditionally on $X=x$, write $H=g(x)^{-1/2}\xi$ with $\xi\sim\Normal(0,I_d)$. Gaussian concentration bounds the infeasibility probability as
    \[
        \mu(\cE^c)
        \leq
        \Prob\bpar{\norm{\xi}\geq\frac{\sqrt d}{r}}
        \leq
        \exp\bpar{-\Omega(d/r^2)}\,.
    \]
    Take the small implicit constant in the choice of $r$ so that $r\leq a_0/[4(2+\gamma)\,\ell^{3/2}]$.
    Let $\Xi$ be the region in Lemma~\ref{lem:acceptance_rate_control} corresponding to $a=a_0$ and $\tau=\tau_0$, and use the acceptance and overlap bounds with $k=\ell$.
    Together with the infeasibility estimate above, we have
    \begin{align*}
        \pi_\theta(\Xi^c)
        &\leq
        \exp\bpar{-\Omega(\ell)}\,,\\
        \dtv(\cT_x,\cT_y)
        &\leq 25/28
        \quad\text{whenever }x,y\in\Xi
        \text{ and }
        \norm{x-y}_{g(x)}\le \step/10\,.
    \end{align*}

    Set $\Delta \defn \frac{\step}{10}$ and $\vartheta \defn 1\wedge\frac{\Delta}{2\sqrt{\bar\nu}} \asymp \frac{r}{\sqrt{d\bar\nu}}$.
    By $\bar\nu$-symmetry and Fact~\ref{fact:llv-nu-symmetry}, for all distinct $x,y\in\interior\K$,
    \[
        \norm{x-y}_{g(x)}
        \wedge
        \norm{x-y}_{g(y)}
        \leq
        \sqrt{\bar\nu}\,d_{\K}(x,y)\,.
    \]
    Increasing the implicit constant in the choice of $\ell$, if
    necessary, ensures that
    \[
        \zeta
        \defn
        \pi_\theta(\Xi^c)
        \leq
        \frac{
            \veps^2\vartheta
        }{
            32\,\bpar{1+\chi^2(\nu_0\mmid\pi_\theta)}
        }\,.
    \]
    For $s \defn \frac{\veps^2}{4\,(1+\chi^2(\nu_0\mmid\pi_\theta))}$, the bound on $\zeta$ ensures that $s\geq8\zeta/\vartheta$. Hence,
    Lemma~\ref{lem:generic-warm-start-mixing} with $\delta=25/28$ yields $\Phi_s(\widetilde{\cT})\gtrsim\vartheta$.
    Substituting this bound into Proposition~\ref{prop:mixing-from-conductance} shows that the lazy $\dw$ has TV-distance at most $\veps$ after
    \[
        O\bpar{
            \vartheta^{-2}
            \log\frac{1\vee\chi^2(\nu_0\mmid\pi_\theta)}{\veps^2}
        }
        =
        O\bpar{
            \frac{d\bar\nu}{r^2}
            \log\frac{1\vee\chi^2(\nu_0\mmid\pi_\theta)}{\veps^2}
        }\quad\text{steps}\,.
    \]
    The claim follows from $r^{-2}\asymp\ell^3\asymp L^3$.
\end{proof}

\subsection{Proof of Lemma~\ref{lem:acceptance_rate_control}}
\label{subsec:generic-acceptance-control}

We prove the two conclusions separately. To bound $\pi_\theta(\Xi^c)$, we
separate infeasible proposals and reduce the feasible bad-proposal event to
an initial derivative moment. To prove transition overlap, we lower-bound the
proposal mass accepted from both points on a common downhill half-space.

We first represent the negative log proposal-density ratio along the
proposal path.
For $t\in[0,\step]$ and every $(x,h)$ such that
$x+th\in\interior\K$, define
\[
    \Theta_t(x,h)
    \defn
    -\frac12\log\det g(x+th)
    +\frac12\,h\tp g(x+th)h\,,
    \qquad
    \dot\Theta_t(x,h)
    \defn
    \partial_t\Theta_t(x,h)\,.
\]
By the proposal density in \eqref{eq:proposal-density}, on $\cE$,
\begin{equation}\label{eq:log-proposal-density-ratio}
    \log\frac{p_{x+\step h}(x)}{p_x(x+\step h)}
    =
    \Theta_0(x,h)-\Theta_\step(x,h)\,.
\end{equation}
Consequently, under $\mu$,
\[
    \mathcal B_{\theta,a}\cap\cE
    =
    \{
        (x,h)\in\mathcal D_\theta\cap\cE:
        \Theta_\step(x,h)-\Theta_0(x,h)>a
    \}\,.
\]
For $k\geq1$, define the $L^k$-norm of the derivative at time $t=0$ by
\[
    R_k
    \defn
    \norm{\dot\Theta_0}_{L^k(\mu)}
    =
    \bpar{
        \E_\mu\bbrack{\abs{\dot\Theta_0(X,H)}^k}
    }^{1/k}\,.
\]
We now bound the derivative moment at each time $t$ by $R_k$ and then control the tail probability.

\begin{lemma}[Tail bound of the bad event]\label{lem:acceptance-to-initial}
Fix $a>0, \step>0$, $k\geq1$, and let $g$ be $C^1$ on $\interior\K$. Then,
\[
    \mu(\mc B_{\theta, a}\cap \mc E) = \mu\bpar{
        \{
            (x,h)\in\mathcal D_\theta\cap\cE:
            \Theta_\step(x,h)-\Theta_0(x,h)\geq a
        \}
    }
    \leq
    e^a\,\bpar{\frac{\step R_k}{a}}^k \,.
\]
\end{lemma}

\begin{proof}
    For $t\in[0,\step]$, define
    $\cE_t \defn \{(x,h)\in\interior\K\times\Rd: x+th\in\interior\K\}$.
    Note that $\cE_\step=\cE$.
    Recall that the joint law $(X,H)\sim\mu$ in \eqref{eq:joint-law} has the following density on $\interior\K\times\Rd$:
    \[
        \D\mu(x,h)
        =
        \frac{1}{Z_\theta(2\pi)^{d/2}}\, \exp\bpar{-\inner{\theta,x}-\Theta_0(x,h)} \,\D x\D h\,.
    \]
    For $t\in[0,\step]$, define a measure on
    $\mathcal D_\theta\cap\cE_t$ by
    \[
        \D \widetilde{\mu}_t(x,h)
        \defn
        e^{-\{\Theta_t(x,h)-\Theta_0(x,h)\}}\,\D\mu(x,h)\,,
    \]
    and set it to zero outside this set. For every nonnegative measurable
    $\varphi$, the change of variables $y=x+th$ yields
    \begin{align*}
        &\int \varphi(x+th,h)\, \D \widetilde{\mu}_t(x,h)\\
        &\quad=
        \frac{1}{Z_\theta\,(2\pi)^{d/2}}
        \int
        \varphi(y,h)\,
        \Ind_{\interior\K}(y)\,\Ind_{\interior\K}(y-th)\,
        \Ind_{\sbrace{\inner{\theta,h}\leq0}}(h)
        \exp\bpar{
            -\inner{\theta,y}
            +t\inner{\theta,h}
            -\Theta_0(y,h)
        }\,\D y \D h\\
        &\quad\leq
        \int \varphi(y,h)\,\D\mu(y,h)\,,
    \end{align*}
    where the inequality follows from $t\inner{\theta,h}\leq0$ on $\mathcal D_\theta$. Since $\dot\Theta_t(x,h)=\dot\Theta_0(x+th,h)$, choosing $\varphi(y,h)=\abs{\dot\Theta_0(y,h)}^k$ yields
    \[
        \int\babs{\dot\Theta_t(x,h)}^k \,\D \widetilde{\mu}_t(x,h)
        \le R_k^k\,.
    \]

    For $\mu$-almost every $(x,h)\in\cE$, convexity of $\K$ implies
    $(x,h)\in\cE_t$ for every $t\in[0,\step]$. We integrate the
    absolute derivative only over times when the increment is at most $a$.
    For $(x,h)\in\mathcal D_\theta\cap\cE$, define
    \[
        J_a(x,h)
        \defn
        \int_0^\step
        \babs{\dot\Theta_t(x,h)}\,
        \Ind_{\{\Theta_t(x,h)-\Theta_0(x,h)\leq a\}}\,\D t\,,
    \]
    and set $J_a(x,h)=0$ otherwise. For such $(x,h)$, suppose that $\Theta_\step(x,h)-\Theta_0(x,h)\geq a$, and define the first hitting time by $t_\star \defn \inf\{t\in[0,\step]: \Theta_t(x,h)-\Theta_0(x,h)\geq a\}$.
    Continuity gives $\Theta_{t_\star}(x,h)-\Theta_0(x,h)=a$ and
    $\Theta_t(x,h)-\Theta_0(x,h)\leq a$ for $0\leq t\leq t_\star$.
    Therefore,
    \[
        a
        =
        \Theta_{t_\star}(x,h)-\Theta_0(x,h)
        =
        \int_0^{t_\star} \dot\Theta_t(x,h)\, \D t
        \leq
        J_a(x,h)\,.
    \]

    On $\mathcal D_\theta\cap\cE$,
    $\D\mu=e^{\Theta_t-\Theta_0}\,\D \widetilde{\mu}_t$. At the times
    included in $J_a$, this factor is at most $e^a$. By Jensen's inequality,
    \begin{align*}
        \E_{\mu}[J_a^k]
        &\le
        \step^{k-1}\int_0^\step\int
        \babs{\dot\Theta_t}^k\,
        \Ind_{\mathcal D_\theta\cap\cE}(x,h)\,
        \Ind_{\{\Theta_t(x,h)-\Theta_0(x,h)\leq a\}}\,
        \D\mu\D t\\
        &=
        \step^{k-1}\int_0^\step\int
        \babs{\dot\Theta_t}^k\,
        \Ind_{\mathcal D_\theta\cap\cE}(x,h)\,
        \Ind_{\{\Theta_t(x,h)-\Theta_0(x,h)\leq a\}}\,
        e^{\Theta_t-\Theta_0}
        \,\D\widetilde{\mu}_t\,\D t\\
        &\le
        e^a\step^{k-1}\int_0^\step\int
        \babs{\dot\Theta_t}^k \,\D \widetilde{\mu}_t\D t
        \le e^a\step^kR_k^k\,.
    \end{align*}
    Applying Markov's inequality, we have
    \[
        \mu\bpar{
            \{
                (x,h)\in\mathcal D_\theta\cap\cE:
                \Theta_\step(x,h)-\Theta_0(x,h)\geq a
            \}
        }
        \leq
        \mu(J_a\geq a)
        \leq
        a^{-k}\,\E_\mu[J_a^k]
        \leq
        e^a\bpar{\frac{\step R_k}{a}}^k\,,
    \]
    which completes the proof.
\end{proof}

It remains to bound $R_k$. Whenever $x+th\in\interior\K$, differentiation along the path yields
\begin{equation}
    \dot\Theta_t(x,h)
    =
    \frac12\, \Bpar{
        \Dd g(x+th)[h,h,h]
        -
        \Tr\bpar{
            g(x+th)^{-1}\Dd g(x+th)[h]
        }
    }\,.
    \label{eq:theta-derivative}
\end{equation}

\begin{lemma}[$R_k$-moment]\label{lem:generic-Rk-bound}
Let $g$ be SSC and have $\gamma$-bounded mixed $23$-trace on
$\interior\K$ for some $\gamma\geq0$. Then, for every $k\ge2$,
\begin{equation}
    R_k
    \le
    \sqrt6\,(k-1)^{3/2}\sqrt d
    +\gamma\sqrt{kd}
    \le
    2k^{3/2}\sqrt d\,(2+\gamma)\,.
    \label{eq:generic-Rk-bound}
\end{equation}
\end{lemma}

\begin{proof}
    Fix $x\in\interior\K$, write $g=g(x)$ and $\cJ=\cJ_x$, and let
    $H\sim\Normal(0,g^{-1})$. For
    $\xi\defn g^{1/2}H\sim\Normal(0,I_d)$, the derivative
    identity~\eqref{eq:theta-derivative} becomes
    \begin{equation}\label{eq:dotTheta-tensor-form}
        2\dot\Theta_0(x,H)
        =
        \cJ[\xi,\xi,\xi] - \inner{\Tr_{12}\cJ,\xi}\,.
    \end{equation}
    Using $\cJ_{ijk}=\cJ_{jik}$, define the symmetrization of $\cJ$ by
    $\overline{\cJ}_{ijk} \defn \frac13\, (\cJ_{ijk}+\cJ_{ikj}+\cJ_{jki})$.
    It satisfies
    \begin{equation}
        3\Tr_{12}\overline{\cJ}
        =
        \Tr_{12}\cJ
        +
        2\Tr_{23}\cJ \,.
        \label{eq:symmetrized-partial-trace}
    \end{equation}

    Define the third-order Hermite polynomial
    \[
        \mathsf H_{\overline{\cJ}}(\xi)
        \defn
        \sum_{ij\ell}
        \overline{\cJ}_{ij\ell}\,
        (\xi_i\xi_j\xi_\ell
            -\delta_{ij}\xi_\ell
            -\delta_{i\ell}\xi_j
            -\delta_{j\ell}\xi_i)\,,
    \]
    where $\delta_{ij}=1$ when $i=j$ and $0$ otherwise.
    Using $\overline{\cJ}[\xi,\xi,\xi]=\cJ[\xi,\xi,\xi]$, we have
    \[
        \cJ[\xi,\xi,\xi]
        =
        \mathsf H_{\overline{\cJ}}(\xi)
        +
        3\,\inner{\Tr_{12}\overline{\cJ},\xi}
        \overset{\eqref{eq:symmetrized-partial-trace}}{=}
        \mathsf H_{\overline{\cJ}}(\xi) +
        \inner{\Tr_{12}\cJ,\xi}
        +
        2\,\inner{\Tr_{23}\cJ,\xi}\,.
    \]
    Substituting this decomposition into \eqref{eq:dotTheta-tensor-form} yields
    \begin{equation}\label{eq:dotTheta-chaos-decomposition}
        2\dot\Theta_0(x,H)
        =
        \mathsf H_{\overline{\cJ}}(\xi)
        +
        2\,\inner{\Tr_{23}\cJ,\xi}\,.
    \end{equation}
    
    We bound $\mathsf H_{\overline{\cJ}}(\xi)$ using SSC and $\inner{\Tr_{23}\cJ,\xi}$ using the mixed $23$-trace assumption.
    Since symmetrization is an orthogonal projection in the Frobenius inner
    product, $\norm{\overline{\cJ}}_\F\leq\norm{\cJ}_\F$. By SSC,
    $\vecnorm{\cJ}{\sbrace{12}\sbrace{3}}\leq2$. Applying this bound to
    each coordinate direction and summing, we obtain
    \[
        \norm{\cJ}_{\F}^2
        =
        \sum_{\ell=1}^d
        \sum_{i,j=1}^d\cJ_{ij\ell}^2
        \le
        4d\,.
    \]
    By a direct Gaussian moment calculation, $\E\bbrack{\mathsf H_{\overline{\cJ}}(\xi)^2} = 3!\,\norm{\overline{\cJ}}_\F^2$.
    By Gaussian hypercontractivity~\cite[Theorem~5.8]{janson_gaussian_1997},
    \[
        \bpar{
            \E\bbrack{\abs{\mathsf H_{\overline{\cJ}}(\xi)}^k}
        }^{1/k}
        \le
        (k-1)^{3/2}
        \bpar{\E[\mathsf H_{\overline{\cJ}}(\xi)^2]}^{1/2}
        =
        \sqrt6\,(k-1)^{3/2}\,
        \norm{\overline{\cJ}}_\F
        \le
        2\sqrt6\,(k-1)^{3/2}\sqrt d\,.
    \]
    Applying the same bound to the linear term and using $\sqrt{k-1}\leq\sqrt k$, we obtain
    \[
        \bpar{\E\bbrack{\abs{\inner{\Tr_{23}\cJ,\xi}}^k}}^{1/k}
        \le
        \sqrt{k}\, \norm{\Tr_{23}\cJ}
        \le
        \gamma\sqrt{kd}\,.
    \]
    Applying the triangle inequality to \eqref{eq:dotTheta-chaos-decomposition} in $L^k$, we obtain, for every
    $x\in\interior\K$,
    \[
        \bpar{
            \E_{\Normal(0,g(x)^{-1})}
            \bbrack{\abs{\dot\Theta_0(x,H)}^k}
        }^{1/k}
        \leq
        \sqrt6\,(k-1)^{3/2}\sqrt d
        +
        \gamma\sqrt{kd}\,.
    \]
    Integrating the conditional $k$-th-moment bound over $X\sim\pi_\theta$ proves the first inequality in \eqref{eq:generic-Rk-bound}. The second follows from $k\geq2$.
\end{proof}

\begin{proof}[Proof of Lemma~\ref{lem:acceptance_rate_control}]
    We first bound $\pi_\theta(\Xi^c)$.
    By Lemmas~\ref{lem:acceptance-to-initial} and~\ref{lem:generic-Rk-bound},
    \[
        \mu(\mathcal B_{\theta,a})
        \leq
        e^a\,\bpar{\frac{\step R_k}{a}}^k
        +
        \mu(\cE^c)
        \leq
        e^a\,
        \bpar{
            \frac{
                2r(2+\gamma)\,k^{3/2}
            }{a}
        }^k
        +
        \mu(\cE^c)\,.
    \]
    By Markov's inequality and the tower property,
    \begin{align*}
        \pi_\theta(\Xi^c)
        &=
        \Prob_{X\sim\pi_\theta}
        \bpar{
            \Prob_\mu(\mathcal B_{\theta,a}\mid X)>\tau
        }\leq
        \frac{
            \E_{\pi_\theta}[\Prob_\mu(\mathcal B_{\theta,a}\mid X)]
        }{\tau}
        =
        \frac{\mu(\mathcal B_{\theta,a})}{\tau}\\
        &\leq
        \frac{e^a}{\tau}\,
        \bpar{
            \frac{2r(2+\gamma)k^{3/2}}{a}
        }^k
        +
        \frac{\mu(\cE^c)}{\tau}\,.
    \end{align*}

    We next prove transition overlap. For $u\in\Xi$, let
    $\mathcal B_{\theta,a}(u)\defn \{z\in\Rd: (u,\frac{z-u}{\step})\in\mathcal B_{\theta,a}\}$ be the set of bad proposal endpoints from $u$.
    Then,
    $\proposal_u [\mathcal B_{\theta,a}(u)]
    =\Prob_{\Normal(0,g(u)^{-1})}
    [(u,H)\in\mathcal B_{\theta,a}]\leq\tau$.

    Fix distinct $x,y\in\Xi$ satisfying $\norm{x-y}_{g(x)}\leq\step/10$.
    By Lemma~\ref{lem:proposal_overlap}, $\dtv(\proposal_x,\proposal_y)\leq1/4$.
    Since the proposal-overlap bound is symmetric, we may relabel $x$ and
    $y$ so that $\inner{\theta,x}\leq\inner{\theta,y}$.
    For $D_x^- :=\{z \in \Rd : \inner{\theta , z}\le \inner{\theta , x}\}$,
    note that $\inner{\theta,z}\leq\inner{\theta,x} \leq\inner{\theta,y}$ for every $z\in D_x^-$.
    By Gaussian symmetry, $\proposal_x(D_x^-)  =  1/2$. Therefore,
    \[
        \int_{D_x^-}\min\{p_x(z),p_y(z)\}\,\D z
        \geq
        \proposal_x(D_x^-)
        -\dtv(\proposal_x,\proposal_y)
        \geq
        \frac14\,.
    \]

    Let $q_u(z) \defn p_u(z)\,\alpha_g(u,z)$ be the density of accepted moves from $u$.
    Then, $\cT_u(A)\geq\int_A q_u(z)\,\D z$ for every measurable $A\subseteq\Rd$.
    For $u\in\{x,y\}$ and $z \in D_x^-\setminus\mathcal B_{\theta,a}(u)$, the proposal is
    feasible and downhill from $u$, and
    $p_z(u)/p_u(z)\geq e^{-a}$. In the acceptance ratio
    \eqref{eq:exponential-acceptance}, the target-density factor is at least
    one for a downhill proposal. Hence $\alpha_g(u,z)\geq e^{-a}$ and
    $q_u(z)\geq e^{-a}p_u(z)$. It follows that
    \begin{align*}
        1-\dtv(\cT_x,\cT_y)
        &\geq
        \int_{D_x^-\setminus
        [\mathcal B_{\theta,a}(x)\cup\mathcal B_{\theta,a}(y)]}
        \min\{q_x(z),q_y(z)\}\,\D z\\
        &\geq
        e^{-a}\, \Bpar{
            \int_{D_x^-}\min\{p_x(z),p_y(z)\}\,\D z
            -\proposal_x\bpar{\mathcal B_{\theta,a}(x)}
            -\proposal_y\bpar{\mathcal B_{\theta,a}(y)}
        }\\
        &\geq
        e^{-a}\,\bpar{\frac14-2\tau}\,.
    \end{align*}
    Rearranging gives the claimed transition-overlap bound.
\end{proof}

\section{Proof of Theorem~\ref{theorem:generic-pointwise-mixing} via ordinary conductance}\label{sec:spectral-analysis}

In the previous section, we constructed a high-probability region $\Xi$ on which every point has high acceptance probability and established one-step coupling among nearby points in this region.
We now prove Theorem~\ref{theorem:generic-pointwise-mixing}, showing that high acceptance probability and one-step coupling can be extended throughout $\intk$ under additional regularity assumptions on the local metric $g$.

The proof consists of three steps. First, Lemma~\ref{lem:fourth-order-bootstrap-ASC} uses SSC, bounded mixed $23$-trace, and fourth-order bootstrap control to establish ASC.
Second, proposal overlap from SSC combines with LTSC and ASC to establish transition overlap for nearby points throughout $\interior\K$. Third,
$\bar\nu$-symmetry and cross-ratio isoperimetry convert this overlap into an ordinary conductance lower bound (not just $s$-conductance). Proposition~\ref{prop:mixing-from-conductance} then yields $\chi^2$-contraction.

The second and third steps follow standard Dikin-walk arguments \cite{KV24ipm,JC25regularized}. The additional step is Lemma~\ref{lem:fourth-order-bootstrap-ASC}, which uses SSC and bounded mixed $23$-trace to control the initial cubic term and fourth-order bootstrap control to bound the integral remainder in Taylor's formula along a Gaussian proposal path.

\subsection{A fourth-order bootstrap criterion for ASC and one-step coupling}

We first control the initial derivative and the integral remainder along a Gaussian proposal path.

\begin{lemma}[Fourth-order bootstrap criterion for ASC]\label{lem:fourth-order-bootstrap-ASC}
If $g$ is SSC, has bounded mixed $23$-trace, and has fourth-order
bootstrap control, then $g$ is ASC. 
\end{lemma}

\begin{proof}
    Fix $\rho\in(0,1/2)$ and $x\in\interior\K$. We verify the event in Definition~\ref{def:asc} with a radius that is uniform in $x$ and $d$.
    Let $\xi\sim\Normal(0,I_d)$, and set $h=g(x)^{-1/2}\xi$ and $x_t=x+th$.
    Whenever $x_t\in\interior\K$, set
    $\xi_t=g(x_t)^{1/2}h$ and $F(t)=h\tp g(x_t)h$.
    For every feasible $\step\geq0$ and $Y=x+\step h$,
    \begin{equation}\label{eq:F-taylor}
        \abs{
            \norm{Y-x}_{g(Y)}^2-
            \norm{Y-x}_{g(x)}^2
        }
        =
        \step^2\,\abs{F(\step)-F(0)}\,.
    \end{equation}
    Thus, it suffices to control the variation of $F$ along the proposal path. On every feasible segment, differentiation yields $F'(t)=\cJ_{x_t}^g[\xi_t^{\otimes3}]$ and $F''(t)=\cQ_{x_t}^g[\xi_t^{\otimes4}]$.
    Hence, for every $\step\geq0$ such that $\{x_t:0\leq t\leq\step\}\subseteq\interior\K$, Taylor's formula yields
    \begin{equation}
        F(\step)-F(0)
        =
        \step\,\cJ_x^g[\xi^{\otimes3}]
        +
        \int_0^\step
        (\step-t)\,\cQ_{x_t}^g[\xi_t^{\otimes4}]\,\D t\,.
    \label{eq:Gaussian-ray-Taylor}
    \end{equation}
    
    We bound the initial derivative and the integral remainder separately.
    Write $\cJ=\cJ_x^g$. The Gaussian decomposition from the proof of
    Lemma~\ref{lem:generic-Rk-bound} expresses the initial derivative as
    \[
        F'(0)
        =
        \mathsf H_{\overline{\cJ}}(\xi)
        +\inner{
            \Tr_{12}\cJ+2\Tr_{23}\cJ,
            \xi
        }\,,
    \]
    where $\overline{\cJ}$ and $\mathsf H_{\overline{\cJ}}$ are, respectively, the
    symmetrization and the third-order Hermite polynomial defined there. SSC
    controls the Hermite term, Remark~\ref{rmk:regularity-ssc} controls
    $\Tr_{12}\cJ$, and bounded mixed $23$-trace controls
    $\Tr_{23}\cJ$. Applying the moment bounds from that proof with
    $k=\lceil2+\log(2/\rho)\rceil$ and then Markov's inequality shows that, with probability at least $1-\rho/2$,
    \[
        \abs{F'(0)}
        \leq
        M_{3,\rho}\sqrt d
        \qquad\text{for }
        M_{3,\rho}
        \asymp
        \bpar{1+\log\frac{2}{\rho}}^{3/2}\,.
    \]

    For the integral remainder, apply
    Definition~\ref{def:fourth-order-bootstrap-control} with
    $\delta=\rho/2$. On its event, which has probability at least
    $1-\rho/2$, feasibility in condition~\textup{(B1)} keeps the path in
    $\interior\K$ up to $r_{\rho/2}/\sqrt d$. The initial bound in
    condition~\textup{(B1)}, the integral inequality in
    condition~\textup{(B2)}, and Gr\"onwall's inequality imply
    \[
        \omega_t
        \leq
        \omega_0\exp(\eta_{\rho/2}\sqrt d\,t)
        \leq
        c_{\rho/2}
        \exp(\eta_{\rho/2}\,r_{\rho/2})
        \qquad
        \text{for }0\leq t\leq r_{\rho/2}/\sqrt d\,.
    \]
    Combining this estimate with the fourth-order bound in
    condition~\textup{(B3)}, we obtain
    \[
        \abs{\cQ_{x_t}^g[\xi_t^{\otimes4}]}
        \leq
        M_{4,\rho/2}d
        \qquad\text{for }
        M_{4,\rho/2}
        \defn
        b_{\rho/2}c_{\rho/2}
        \exp(\eta_{\rho/2}\,r_{\rho/2})\,.
    \]
    The intersection of the two events has probability at least $1-\rho$.
    Define
    \[
        r_\rho^\star
        \defn
        \min\bbrace{
            r_{\rho/2},
            \frac{\rho}{M_{3,\rho}},
            \bpar{\frac{2\rho}{1\vee M_{4,\rho/2}}}^{1/2}
        }\,.
    \]
    The three terms in this minimum respectively ensure path feasibility,
    an initial-derivative contribution at most $\rho$, and an integral
    remainder at most $\rho$ in Taylor's formula~\eqref{eq:Gaussian-ray-Taylor}.
    For any $0<r\leq r_\rho^\star$, set
    $\step\defn r/\sqrt d$ and $Y=x+\step h$.
    Then $Y\sim\Normal(x,r^2g(x)^{-1}/d)$. On the intersection of the two events,
    $r\leq r_{\rho/2}$, so condition~\textup{(B1)} ensures that the segment
    $\{x_t:0\leq t\leq\step\}$ is feasible. We may therefore apply Taylor's
    formula~\eqref{eq:Gaussian-ray-Taylor}; combining the preceding bounds with \eqref{eq:F-taylor} yields
    \[
        \abs{
            \norm{Y-x}_{g(Y)}^2-\norm{Y-x}_{g(x)}^2
        }
        =\eta^2\,\abs{F(\eta)-F(0)} 
        \leq
        \eta^2\, \bpar{M_{3,\rho}r + \frac{M_{4,\rho/2}r^2}{2}}
        \leq
        \frac{2\rho r^2}{d}\,.
    \]
    These feasibility and increment bounds hold with probability at least
    $1-\rho$. Moreover, $M_{3,\rho}$ and the parameters in
    Definition~\ref{def:fourth-order-bootstrap-control} at $\delta=\rho/2$
    depend only on $\rho$ and universal constants. Hence so does
    $r_\rho^\star$. Since $x$ was arbitrary, taking
    $r_\rho^{\mathrm{ASC}}=r_\rho^\star$ proves the claim.
\end{proof}

We next combine proposal overlap with LTSC and ASC to obtain transition overlap for nearby points throughout $\interior\K$. The following result essentially appeared in previous work~\cite{KV24ipm,JC25regularized}, and we defer its streamlined proof to \S\ref{sec:one_step_coupling_under_ssc_ltsc_and_asc}.

\begin{lemma}[Pointwise transition overlap]
\label{lem:pointwise-Dikin-framework}
Let $g:\interior\K\to\mathbb S_{++}^d$ be SSC, LTSC, and ASC.
Set $\rho_\star \defn \frac1{64}$ and $r_0 \defn 1\wedge r_{\rho_\star}^{\mathrm{ASC}}$, where $r_{\rho_\star}^{\mathrm{ASC}}>0$ is a radius given by Definition~\ref{def:asc}.
There are universal constants $c_0\in(0,1/2]$ and $\delta_0\in(0,1)$ such that for every exponential target $\pi_\theta$, the Dikin kernel with proposal radius $r=c_0r_0$ satisfies
\[
    \dtv(\cT_x,\cT_y)
    \leq
    \delta_0
    \qquad\text{whenever}\qquad
    \norm{x-y}_{g(x)}
    \leq
    \frac{\step}{10}
    =
    \frac{c_0r_0}{10\sqrt d}\,.
\]
\end{lemma}

\subsection{Proof of Theorem~\ref{theorem:generic-pointwise-mixing}}

We now derive ordinary conductance and $\chi^2$-contraction from pointwise
transition overlap.

\begin{proof}[Proof of Theorem~\ref{theorem:generic-pointwise-mixing}]
By Lemma~\ref{lem:fourth-order-bootstrap-ASC}, SSC, bounded mixed
$23$-trace, and fourth-order bootstrap control imply ASC; at $\rho_\star=1/64$, the corresponding ASC radius is universal.
Together with SSC and LTSC, Lemma~\ref{lem:pointwise-Dikin-framework} then yields a proposal radius $r\in(0,1/2]$ and a universal constant $\delta_0\in(0,1)$.
Then, $\dtv(\cT_x,\cT_y)\leq\delta_0$ whenever
$\norm{x-y}_{g(x)}\leq\Delta \defn  \frac{r}{10\sqrt d} \asymp \frac{1}{\sqrt d}$.

The kernel $\cT$ is $\pi_\theta$-reversible, SSC implies self-concordance, and $\pi_\theta(\interior\K)=1$. Moreover, $\bar\nu$-symmetry and Fact~\ref{fact:llv-nu-symmetry} imply $\norm{x-y}_{g(x)} \wedge \norm{x-y}_{g(y)} \leq \sqrt{\bar\nu}\,d_\K(x,y)$ for all distinct $x,y\in\interior\K$.
Since $\Delta\leq1$, the preceding bounds verify the hypotheses of
Lemma~\ref{lem:generic-warm-start-mixing} with $\mathsf T=\cT$,
$\Xi=\interior\K$, $\zeta=s=0$, $\delta=\delta_0$, and the displayed
$\Delta$. Therefore,
\[
    \Phi_0(\widetilde{\cT})
    \gtrsim
    \frac{r}{\sqrt{d\bar\nu}}
    \asymp
    \frac{1}{\sqrt{d\bar\nu}}\,.
\]
Substituting this conductance bound into Proposition~\ref{prop:mixing-from-conductance} yields, for a universal constant $c>0$,
\[
    \chi^2(\nu_0\widetilde{\cT}^{\,N}\mmid\pi_\theta)
    \leq
    \exp\bpar{-\frac{cN}{d\bar\nu}}\,\chi^2(\nu_0\mmid\pi_\theta)\,.
\]
Solving this inequality for $N$ gives the stated $\chi^2$ mixing-time bound and completes the proof.
\end{proof}

\section{Applications}\label{sec:applications}\label{sec:tv-applications}

This section proves the four corollaries stated in \S\ref{subsec:consequences}. In \S\ref{sec:polytope-mixing}, we prove Corollary~\ref{cor:generic-weighted-Dikin-mixing-polytope} by reducing the assumptions of
Theorem~\ref{theorem:generic-weighted-Dikin-mixing} to three conditions for general weighted Dikin metrics on polytopes.
\S\ref{sec:LS-verification} verifies these conditions for the Lee--Sidford, Lewis-weight, and John metrics and proves Corollary~\ref{cor:LS-mixing}.
\S\ref{ssub:dimension_square_mixing_for_psd_cones} verifies the conditions of Theorem~\ref{theorem:generic-weighted-Dikin-mixing} for a hybrid metric combining the Hessian of the positive-semidefinite barrier with an LS component and proves Corollary~\ref{cor:dimension-square-SDP}.
Finally, \S\ref{sec:spectral-LS-verification} verifies the conditions in Theorem~\ref{theorem:generic-pointwise-mixing} for a scaled LS metric and proves Corollary~\ref{cor:spectral-LS-mixing}.

\subsection{Mixing framework for weighted Dikin walks on polytopes}\label{sec:polytope-mixing}

The proof has two steps. The row-derivative operator bound in condition~(i) yields SSC and bounded mixed $23$-trace, while the ellipsoid-containment and weight-trace bounds in conditions~(ii) and~(iii) yield $\bar\nu$-symmetry.

\begin{proof}[Proof of Corollary~\ref{cor:generic-weighted-Dikin-mixing-polytope}]
    Fix $x\in\interior\K$ and suppress the dependence on $x$ when it is clear from context.

    \paragraph{SSC and mixed $23$-trace.}
    For $B = W^{1/2}A_x$, set $U\defn Bg^{-1/2}$ and $P\defn UU\tp$.
    Then, $U\tp U=I_d$, $P$ is an orthogonal projection, and
    $\Sigma=\Diag  P$. For $h=g^{-1/2}z$, differentiating
    $B=W^{1/2}A_x$ gives
    \[
        \Dd B[h]
        =
        \Diag(Rz)\,B\,.
    \]
    Since $g=B\tp B$, differentiation yields the key identity
    \begin{equation}\label{eq:weighted-Dikin-tensor-factorization}
        g^{-1/2}\,\Dd g[g^{-1/2}z]\,g^{-1/2}
        =
        2U\tp \Diag(Rz)U\,.
    \end{equation}
    Here, $\circ$ denotes the Hadamard product and $P^{\circ2}\defn P\circ P$.  
    As $\Sigma-P^{\circ2}=P\circ(I_m-P)\succeq0$,
    \[
        \frac14\,
        \norm{g^{-1/2}\,\Dd g[g^{-1/2}z]\,g^{-1/2}}_{\F}^2
        =
        (Rz)\tp P^{\circ2}(Rz)
        \leq
        \norm{\Sigma^{1/2}Rz}^2
        \leq
        C^2\,\norm z^2
        \quad\text{for  } z\in\Rd\,.
    \]
    Thus, $\vecnorm{\cJ_x^g}{\sbrace{12}\sbrace{3}}\leq2C\leq2$, so $g$ is
    SSC.

    For the mixed trace, let $u_i\tp$ and $r_i\tp$ be the $i$-th rows of
    $U$ and $R$, and set $q_i\defn\inner{u_i,r_i}$.
    Taking the mixed $23$-trace in the
    factorization~\eqref{eq:weighted-Dikin-tensor-factorization} yields
    $\Tr_{23}\cJ_x^g=2U\tp q$. Since
    $\norm{u_i}^2=P_{ii}=\Sigma_{ii}$,
    \begin{equation}\label{eq:tr23-bound}
    \begin{aligned}
        \norm{\Tr_{23}\cJ_x^g}
        &\leq
        2\,\norm q
        \leq
        2\,\Bpar{\sum_{i=1}^m\Sigma_{ii}\,\norm{r_i}^2}^{1/2}
        =
        2\,\norm{\Sigma^{1/2}R}_{\F}\\
        &\leq
        2\sqrt d\,\norm{\Sigma^{1/2}R}_{\op}
        \leq
        2C\sqrt d\,.
    \end{aligned}
    \end{equation}
    The SSC and mixed-trace bounds verify conditions~(i) and~(iii) of
    Theorem~\ref{theorem:generic-weighted-Dikin-mixing}.

    \paragraph{$\bar\nu$-symmetry.}
    The ellipsoid $\msf D_g(x,1)$ is symmetric about $x$, so condition~(ii) of the corollary implies $\msf D_g(x,1)\subseteq\K\cap(2x-\K)$.
    For the outer inclusion, if $x+h\in\K\cap(2x-\K)$, then
    $\norm{A_xh}_\infty\leq1$. Condition~(iii) therefore gives
    \[
        \norm{h}_{g(x)}^2
        =
        \sum_{i=1}^m w_{x,i}\,(a_{x,i}\tp h)^2
        \leq
        \Tr W_x
        \leq
        \bar\nu\,.
    \]
    Hence, $\K\cap(2x-\K)\subseteq\msf D_g(x,\sqrt{\bar\nu})$, so $g$ is $\bar\nu$-symmetric. The claim follows from Theorem~\ref{theorem:generic-weighted-Dikin-mixing}.
\end{proof}

\subsection{Lee--Sidford, Lewis-weight, and John metrics}\label{sec:LS-verification}

We now verify (i)--(iii) of Corollary~\ref{cor:generic-weighted-Dikin-mixing-polytope} for the Lee--Sidford, Lewis-weight, and John metrics in Lemmas~\ref{lem:LS-metric-conditions}--\ref{lem:John-metric-conditions}.
Each proof checks the row-derivative operator bound in (i), the equivalent leverage-score form of (ii), and the weight-trace bound in (iii), in this order. The normalizing factors in these metrics are chosen so that the bound in (i) is at most one.

\begin{lemma}[Lee--Sidford metric]\label{lem:LS-metric-conditions}
The Lee--Sidford metric $g_{\mathrm{LS}}$ in \eqref{eq:LS-metric} satisfies (i)--(iii) of Corollary~\ref{cor:generic-weighted-Dikin-mixing-polytope} with $C=1$ and
$\bar\nu=\bar\nu_{\mathrm{LS}}\defn \Lambda_{\mathrm{LS}}^2m^{2/p_{\mathrm{LS}}}d^{1-2/p_{\mathrm{LS}}}
\lesssim d\log^2 m$.
\end{lemma}

\begin{proof}
    Set $p\defn p_{\mathrm{LS}}$ and $\Lambda\defn\Lambda_{\mathrm{LS}}$, and write $\widehat W_x\defn\Lambda^2\,(W_x^{(p)})^{1-2/p}$.
    Fix $x\in\interior\K$ and abbreviate $g\defn g_{\mathrm{LS}}(x)$, $W\defn W_x^{(p)}$, and $w\defn w_x^{(p)}$. Set $\beta\defn\frac12-\frac1p$ and $U\defn\Lambda W^\beta A_xg^{-1/2}$. 
    Note that $U\tp U=I_d$, $\sum_iw_i=d$, and $0<w_i\leq1$.
    Let $R_x$ be the operator in \eqref{eq:relative-row-derivative} for the effective weight matrix $\widehat W_x$.

    For (i), the derivative identities of Lee and Sidford \cite[Lemmas~24 and~31]{LS19solving} provide a symmetric matrix $N$ such that $0\preceq N\preceq pI$ and
    \[
        \beta\,\Dd\log w_x^{(p)}[h]
        =
        -\beta W^{-1/2}NW^{1/2}A_xh\qquad \text{for every direction } h\,.
    \]
    When $p=2$, we take $N=0$. Since
    $\frac12\,\Dd\log(\diag\widehat W_x)[h]
    =\beta\,\Dd\log w_x^{(p)}[h]$ and
    $1+\beta p=p/2=\Lambda$, for $h=g^{-1/2}z$,
    \[
        W^{1/2}R_xz
        =
        -\frac1\Lambda\,
        (I+\beta N)\,W^{1/p}Uz\,.
    \]
    Hence,
    \begin{equation}\label{eq:LS-operator}
    \norm{W^{1/2}R_x}_{\op}
    \leq
    \frac{1}{\Lambda}\,\norm{I+\beta N}_{\op} \norm{W^{1/p}}_{\op} \norm U_{\op}
    \leq
    1\,.
    \end{equation}
    The identity $\sigma(\widehat W_x^{1/2}A_x)=w$ shows that the leverage-score matrix in Corollary~\ref{cor:generic-weighted-Dikin-mixing-polytope} is $W$.
    Thus, inequality~\eqref{eq:LS-operator} proves condition~(i) with $C=1$.

    For (ii), the same identity, $\Lambda\geq1$, and $0<w_i\leq1$ give $w_i\leq\Lambda^2w_i^{1-2/p}=(\diag\widehat W_x)_i$.
    For (iii), H\"older's inequality and $\sum_iw_i=d$ give
    \[
        \Tr\widehat W_x
        =
        \Lambda^2\sum_{i=1}^m w_i^{1-2/p}
        \leq
        \Lambda^2m^{2/p}d^{1-2/p}
        =
        \bar\nu_{\mathrm{LS}}\,.
    \]
    The choice of $p_{\mathrm{LS}}$ gives
    \[
        \bar\nu_{\mathrm{LS}}
        =
        \frac{dp_{\mathrm{LS}}^2}{4}
        \bpar{\frac md}^{2/p_{\mathrm{LS}}}
        \lesssim
        d\log^2 m\,.
    \]
    Therefore, conditions~(i)--(iii) hold with $C=1$ and the claimed $\bar\nu$.
\end{proof}

We next consider the Lewis-weight metric rather than the powers appearing in the LS metric.

\begin{lemma}[Lewis-weight metric]\label{lem:Lewis-metric-conditions}
The Lewis-weight metric $g_{\mathrm{LW}}$ in \eqref{eq:LW-metric} satisfies conditions~(i)--(iii) of Corollary~\ref{cor:generic-weighted-Dikin-mixing-polytope} with $C=1$ and $\bar\nu=\bar\nu_{\mathrm{LW}}\defn\Lambda_{\mathrm{LW}}^2d \lesssim d\log^2m$.
\end{lemma}

\begin{proof}
    Write $w_x^{\mathrm{LW}}\defn w_x^{(p_{\mathrm{LW}})}$ and
    $W_x^{\mathrm{LW}}\defn W_x^{(p_{\mathrm{LW}})}$.
    Set $\rho_{\mathrm{LW}}\defn2m^{2/(p_{\mathrm{LW}}+2)}$, so
    $\Lambda_{\mathrm{LW}}=(1+p_{\mathrm{LW}}/2)
    \sqrt{\rho_{\mathrm{LW}}}$. Fix $x$ and write
    $W\defn W_x^{\mathrm{LW}}$ and $g_0\defn A_x\tp WA_x$.
    For (i), Lee and Sidford~\cite[Lemmas~24, 26, and~31]{LS19solving} give a symmetric matrix $N$ satisfying
    \begin{equation}\label{eq:Lewis-raw-inputs}
        \Dd\log w_x^{\mathrm{LW}}[h]
        =
        -W^{-1/2}NW^{1/2}A_xh\,,
        \quad
        0\preceq N\preceq p_{\mathrm{LW}}I_m\,,
        \quad
        \max_i
        \frac{\sigma_i(W^{1/2}A_x)}{w_{x,i}^{\mathrm{LW}}}
        \leq
        \rho_{\mathrm{LW}}\,.
    \end{equation}
    Let $U\defn W^{1/2}A_xg_0^{-1/2}$ and $\Sigma\defn\Diag\sigma(W^{1/2}A_x)$. The row-derivative operator for $g_0$ is
    \[
        R_0
        =
        -W^{-1/2}\,\bpar{I_m+\frac12N}\,U\,.
    \]
    By \eqref{eq:Lewis-raw-inputs} and $U\tp U=I_d$,
    \[
        \norm{\Sigma^{1/2}R_0}_{\op}
        \leq
        \sqrt{\rho_{\mathrm{LW}}}
        \,\bpar{1+\frac{p_{\mathrm{LW}}}{2}}
        =
        \Lambda_{\mathrm{LW}}\,.
    \]
    Scaling $g_0$ by $\Lambda_{\mathrm{LW}}^2$ leaves $\Sigma$ unchanged and divides the row-derivative operator by $\Lambda_{\mathrm{LW}}$. Thus, (i) holds with $C=1$.

    For (ii),
    \[
        \sigma_i\bpar{
            (\Lambda_{\mathrm{LW}}^2W)^{1/2}A_x
        }
        =
        \sigma_i(W^{1/2}A_x)
        \leq
        \rho_{\mathrm{LW}}w_{x,i}^{\mathrm{LW}}
        \leq
        \Lambda_{\mathrm{LW}}^2w_{x,i}^{\mathrm{LW}}\,.
    \]
    For (iii), $\sum_iw_{x,i}^{\mathrm{LW}}=d$ gives
    $\Tr(\Lambda_{\mathrm{LW}}^2W)=\Lambda_{\mathrm{LW}}^2d
    =\bar\nu_{\mathrm{LW}}$.
    Our choice of $p_{\mathrm{LW}}$ gives
    $m^{2/(p_{\mathrm{LW}}+2)}=O(1)$ and
    $\Lambda_{\mathrm{LW}}=O(\log m)$, so
    $\bar\nu_{\mathrm{LW}}\lesssim d\log^2m$.
    Thus, (i)--(iii) hold with $C=1$ and the claimed $\bar\nu$.
\end{proof}

We finally turn to the John metric.

\begin{lemma}[John metric]\label{lem:John-metric-conditions}
The John metric $g_{\mathrm{J}}$ in \eqref{eq:J-metric} satisfies (i)--(iii) of Corollary~\ref{cor:generic-weighted-Dikin-mixing-polytope} with $C=1$ and $\bar\nu=\bar\nu_{\mathrm{J}}\defn6(1+\kappa)^2d \lesssim d\log^2\frac{em}{d}$.
\end{lemma}

\begin{proof}
    Fix $x$ and suppress its subscript. Write
    $J\defn A_x\tp ZA_x$, so that
    $g_{\mathrm{J}}(x)=\lambda_{\mathrm{J}}J$. Let
    \[
        P
        \defn
        Z^{\alpha_{\mathrm{J}}/2}A_x\,
        (A_x\tp Z^{\alpha_{\mathrm{J}}}A_x)^{-1}
        A_x\tp Z^{\alpha_{\mathrm{J}}/2}\,,
        \qquad
        \Lambda
        \defn
        \Diag(\diag P)-P^{\circ2}\,,
    \]
    and set $K\defn(Z-\alpha_{\mathrm{J}}\Lambda)^{-1}\,\Lambda$.
    For (i), after translating to our sign convention, the calculus formula of Chen et al.~\cite[Lemma~9(a)]{CDW18MCMC} shows that
    $x\mapsto\zeta_x$ is continuously differentiable and gives the first
    identity below. Chen et al.~\cite[Lemma~10]{CDW18MCMC}, with
    $c_1=0$ and $c_2=1$, give the second:
    \begin{equation}\label{eq:John-derivative-inputs}
        \Dd\log\zeta_x[h]
        =
        -2KA_xh\,,
        \qquad
        K\tp ZK
        \preceq
        \kappa^2Z\,.
    \end{equation}
    Hence, the row-derivative operator for the base metric $J$ is
    $R_0=-(I_m+K)A_xJ^{-1/2}$.

    Let $\Sigma_0\defn\Diag\sigma(Z^{1/2}A_x)$. Chen et al. \cite[Lemma~3(d)]{CDW18MCMC} prove that $\Sigma_0\preceq4Z$.
    Combining this bound with \eqref{eq:John-derivative-inputs} yields
    \[
        \norm{\Sigma_0^{1/2}R_0}_{\op}
        \leq
        2(1+\kappa)\,
        \norm{Z^{1/2}A_xJ^{-1/2}}_{\op}
        =
        2(1+\kappa)\,.
    \]
    Scaling $J$ by $\lambda_{\mathrm{J}}=4(1+\kappa)^2$ divides $R_0$ by $\lambda_{\mathrm{J}}^{1/2}$ and leaves $\Sigma_0$ unchanged, so (i) holds with $C=1$.

    For (ii),
    \[
        \sigma_i\bpar{(\lambda_{\mathrm{J}}Z)^{1/2}A_x}
        =
        \sigma_i(Z^{1/2}A_x)
        \leq
        4\zeta_{x,i}
        \leq
        \lambda_{\mathrm{J}}\zeta_{x,i}\,.
    \]
    For (iii), the identity $\sum_i\sigma_i=d$ gives
    $\sum_{i=1}^m\zeta_{x,i}=d+m\beta_{\mathrm{J}}=3d/2$. Therefore,
    \[
        \Tr(\lambda_{\mathrm{J}}Z)
        =
        6(1+\kappa)^2d
        =
        \bar\nu_{\mathrm{J}}\,.
    \]
    The definition of $\kappa$ gives the stated bound.
    Thus, (i)--(iii) hold with $C=1$ and the claimed $\bar\nu$.
\end{proof}

\begin{proof}[Proof of Corollary~\ref{cor:LS-mixing}]
    Lemmas~\ref{lem:LS-metric-conditions}--\ref{lem:John-metric-conditions} verify the conditions of Corollary~\ref{cor:generic-weighted-Dikin-mixing-polytope} with $C=1$ and the displayed values of $\bar\nu_g$.
    Thus, the corollary gives the universal mixed-$23$-trace constant $\gamma=2$.

    For all three metrics, $\bar\nu_g\lesssim d\log^2(em)$, and hence $1+\sqrt{d\bar\nu_g}\lesssim d\log(em)$. It follows that the quantity
    $L$ in Theorem~\ref{theorem:generic-weighted-Dikin-mixing} satisfies
    $L\lesssim L_g$.
    Substitution into
    Theorem~\ref{theorem:generic-weighted-Dikin-mixing} gives the stated
    bound.
\end{proof}

\subsection{Proof of Corollary~\ref{cor:dimension-square-SDP}: truncated PSD cones}\label{ssub:dimension_square_mixing_for_psd_cones}

Let $n\defn d(d+1)/2$ and define
$\mc A:\mathbb S^d\to\R^m$ by
$(\mc A H)_i\defn\inner{A_i,H}$.  After discarding zero constraint
rows, assume that $\rho\defn\rank\mc A\geq1$.  For
$X\in\interior\K$, let
$S_X\defn\Diag(\mc A X-b)$ and
$\mc A_X\defn S_X^{-1}\mc A$.  Set
\[
    p\defn2\vee\log\frac m\rho\,,
    \qquad
    \Lambda_{\mathrm{LS}}\defn\frac p2\,,
    \qquad
    \beta\defn\frac12-\frac1p\,.
\]
Choose a fixed linear isometry $V:\R^\rho\to\mathbb S^d$ with
$\operatorname{range}(V)=(\ker\mc A)^\perp$. Then, $\mc A_XV$ has full
column rank $\rho$ for every $X\in\interior\K$. Let $W_X$ be the
$\ell_p$-Lewis-weight matrix of $\mc A_XV$. The resulting LS component
is positive semidefinite on $\mathbb S^d$ and has the fixed kernel
$\ker\mc A$. We use the hybrid metric
$g(X)\defn g_{\mathrm{PSD}}(X)+g_{\mathrm{LS}}(X)$, where
\begin{align*}
g_{\mathrm{PSD}}(X)[H_1,H_2]
&\defn
d\,\nabla^2\phi_{\mathrm{PSD}}(X)[H_1,H_2]
=d\Tr(X^{-1}H_1X^{-1}H_2)\,,\\
g_{\mathrm{LS}}(X)[H_1,H_2]
&\defn
\Lambda_{\mathrm{LS}}^2\,
\inner{\mc A_XH_1,W_X^{1-2/p}\mc A_XH_2}\,.
\end{align*}

Applying the proof of Lemma~\ref{lem:LS-metric-conditions} to
$\mc A_XV$ on this fixed $\rho$-dimensional subspace shows that
$g_{\mathrm{LS}}$ is continuously differentiable, satisfies the SSC
derivative bound on $(\ker\mc A)^\perp$, and satisfies the symmetry
estimate for the linear-constraint domain with
\[
    \bar\nu_{\mathrm{LS}}
    \defn
    \Lambda_{\mathrm{LS}}^2m^{2/p}\rho^{1-2/p}
    \lesssim
    \rho\log^2\frac{em}{\rho}
    \leq
    n\log^2(em)\,.
\]

Kook and Vempala prove that $g_{\mathrm{PSD}}$ is SSC and
$d^2$-symmetric
\cite[Corollary~E.24 and Lemma~E.21]{KV24ipm}. Since symmetry
parameters add \cite[Lemma~D.11]{KV24ipm}, the hybrid metric
$g$ is $O(n\log^2m)$-symmetric. We next verify SSC of the sum and the
mixed $23$-trace bound.

Fix $X$, let $g_1\defn g_{\mathrm{PSD}}$ and
$g_2\defn g_{\mathrm{LS}}$, and write $W\defn W_X$ and
$w_{\mathrm{LS}}\defn\diag W_X$. We suppress the dependence on $X$, so $g$, $g_i$,
and $\Dd g_i$ below denote $g(X)$, $g_i(X)$, and $\Dd g_i(X)$,
respectively. Define $T_i\defn g_i^{1/2}g^{-1/2}$ for
$i=1,2$. We first verify SSC of the sum. For $z\in\R^n$, set
$h\defn g^{-1/2}z$ and
\[
    M_i
    \defn
    g_i^{\dagger/2}\,\Dd g_i[h]\,g_i^{\dagger/2}\,,
    \qquad
    T
    \defn
    \begin{bmatrix}T_1\\T_2\end{bmatrix}\,,
    \qquad
    M
    \defn
    \begin{bmatrix}M_1&0\\0&M_2\end{bmatrix}\,,
\]
where $\dagger$ denotes the Moore--Penrose inverse. Here, $g_1$ is positive definite, while $g_2$ has the fixed kernel
$\ker\mc A$. Since $T_1\tp T_1+T_2\tp T_2=I$, we have $T\tp T=I$.
The SSC bounds for $g_1$ and for $g_2$ on its range give $\norm{M_i}_{\F} \leq 2\,\norm{h}_{g_i} = 2\,\norm{T_i z}$.
Consequently, SSC of $g$ follows from
\[
    \norm{g^{-1/2}\,\Dd g[h]\,g^{-1/2}}_{\F}
    =
    \norm{T\tp MT}_{\F}
    \leq
    \norm{M}_{\F}
    \leq
    2\,(\norm{T_1z}^2+\norm{T_2z}^2)^{1/2}
    =
    2\,\norm z\,.
\]

For the mixed-trace calculation, define $\cJ_1[u,v,w] \defn \Dd g_1[g_1^{-1/2}w][g_1^{-1/2}u,g_1^{-1/2}v]$.
Observe that
\begin{align*}
\cJ^g[u,v,w] &= \Dd g_1[g^{-1/2}w][g^{-1/2}u, g^{-1/2}v] + \Dd g_2[g^{-1/2}w][g^{-1/2}u, g^{-1/2}v]\\
&= \cJ_1[T_1u, T_1 v, T_1 w] + \Dd g_2[g^{-1/2}w][g^{-1/2}u, g^{-1/2}v]\,.
\end{align*}

Let $t_i$ be the mixed-$23$ trace of each term. We now bound the norms of $t_1$ and $t_2$, respectively. For $t_1$, let $u$ be a unit vector and let $\{e_i\}$ be an orthonormal basis. Since $\cJ_1$ is fully symmetric,
\[
    \inner{t_1, u} = \sum_{i=1}^n \cJ_1[T_1u, T_1e_i, T_1e_i] = \sum_{i=1}^n \cJ_1[T_1e_i, T_1e_i, T_1u]
    = \inner{\cJ_1[\,\cdot,\cdot,T_1u],T_1T_1\tp}\,.
\]
Since $T_1\tp T_1 + T_2\tp T_2 = I$, we have $T_1\tp T_1 \preceq I$ and $\norm{T_1}_{\op} \le 1$. Using SSC of $g_1$,
\[
    \norm{t_1}\leq \norm{\cJ_1[\,\cdot,\cdot,T_1u]}_{\F} \norm{T_1T_1\tp}_{\F}
    \le 2\,\norm{T_1u} \sqrt{n} \le 2\sqrt{n}\,.
\]

For $t_2$, let $B = \Lambda_{\mathrm{LS}} W^\beta \mc A_X$, $\bar U = Bg^{-1/2}$, and $\bar{R} z\defn \beta\, \Dd \log w_{\mathrm{LS}}[g^{-1/2}z] - \mc A_Xg^{-1/2}z$.
Repeating the factorization in \eqref{eq:weighted-Dikin-tensor-factorization}, we obtain
\[
    g^{-1/2}\,\Dd g_2[g^{-1/2}z]\,g^{-1/2} = 2 \bar U\tp \Diag(\bar R z) \bar U\,.
\]
Applying the mixed-trace calculation used in~\eqref{eq:tr23-bound}
with $\bar\Sigma\defn\Diag[\diag(\bar U\bar U\tp)]$, we obtain
\[
    \norm{t_2} \leq 2\sqrt{n}\,\norm{\bar \Sigma^{1/2} \bar R}_{\op}\,.
\]
Now, define $U = Bg_2^{\dagger / 2}$. Then, $\bar U = UT_2$, so $\bar U \bar U\tp = UT_2 T_2\tp U\tp \preceq UU\tp$. This implies $\bar\Sigma\preceq\Sigma=\Diag[\diag(UU\tp)]$. 
On the range of $g_2$, we have $\bar R=R_2T_2$, where $R_2$ is the row-derivative operator
\[
    R_2z
    \defn
    \beta\,\Dd\log w_{\mathrm{LS}}[g_2^{\dagger/2}z]
    -
    \mc A_Xg_2^{\dagger/2}z\,.
\]
Hence, applying the operator bound~\eqref{eq:LS-operator} to $R_2$ yields
\[
    \norm{t_2}
    \leq
    2\sqrt n\,\norm{\Sigma^{1/2}\bar R}_{\op}
    \leq
    2\sqrt n\sup_{u:\norm{u}=1}\norm{T_2u}
    \leq
    2\sqrt n\,.
\]
Therefore, $\norm{\Tr_{23}\cJ^g}\leq4\sqrt n$, so $g$ has
$4$-bounded mixed $23$-trace. Together with the SSC and symmetry bounds
above, this verifies the three assumptions of
Theorem~\ref{theorem:generic-weighted-Dikin-mixing}. Moreover,
\begin{align*}
    n\bar\nu
    &\lesssim
    n^2\log^2 m 
    \lesssim
    d^4\log^2 m\,,\\
    \log\frac{
        e(1+\sqrt{n\bar\nu})
        \bpar{1+\chi^2(\nu_0\mmid\pi_\Theta)}
    }{\veps^2}
    &\lesssim
    L_g\,.
\end{align*}
Substitution into Theorem~\ref{theorem:generic-weighted-Dikin-mixing} proves Corollary~\ref{cor:dimension-square-SDP}.

\subsection{\texorpdfstring{$\chi^2$-contraction}{Chi-square contraction}
for the Lee--Sidford metric}\label{sec:spectral-LS-verification}

We prove Corollary~\ref{cor:spectral-LS-mixing} by verifying the five conditions of Theorem~\ref{theorem:generic-pointwise-mixing} for the scaled LS metric $\widehat g$.
The main new estimate is $\norm{N_t'}_{\op}\leq6p^2(\omega_t-1)$ in Lemma~\ref{lem:operator-row-map-estimate}. 
Substitution into the standard path estimates establishes bootstrap control~\textup{(B2)} and fourth-order dominance~\textup{(B3)}.

Set
\[
    p\defn p_{\mathrm{LS}} = 2\,(1+\log m)\,,
    \qquad
    \alpha \defn 1-\frac2p\,,
    \qquad
    \beta\defn\frac\alpha2=\frac12-\frac1p\,,
    \qquad
    \lambda\defn\Lambda_{\mathrm{LS}}=\frac p2\,.
\]
Throughout this subsection, write $w_x\defn w_x^{(p)}$ and $W_x\defn W_x^{(p)}$, and define the baseline metric
\[
    g(x)\defn A_x\tp W_x^\alpha A_x\,.
\]
We use $g$ for the main calculations, $\lambda^2g=g_{\mathrm{LS}}$ as the normalized LS metric, and $\widehat g=S\lambda^2g$ as the final metric in
Corollary~\ref{cor:spectral-LS-mixing}.

Lemma~\ref{lem:LS-metric-conditions} verifies conditions~\textup{(i)}--\textup{(iii)} for $g_{\mathrm{LS}}$ with $C=1$ and $\bar\nu=\lambda^2m^{2/p}d^{1-2/p}\leq e\lambda^2d$.
Corollary~\ref{cor:generic-weighted-Dikin-mixing-polytope} therefore shows
that $g_{\mathrm{LS}}$ is SSC, has bounded mixed $23$-trace, and has symmetry parameter at most $e\lambda^2d$.
Convexity of $x\mapsto\log\det g(x)$~\cite[Lemma~4.3]{LLV20strong} and identity~\eqref{eq:logdet-JQ} show that $g$ is LTSC.
Since $\lambda\geq1$, $g_{\mathrm{LS}}$ is also LTSC.
Scaling $g_{\mathrm{LS}}$ up by $S$ preserves SSC, bounded mixed $23$-trace, and LTSC, and multiplies its symmetry parameter by $S$.
Hence, $\widehat g$ satisfies \textup{(i)}--\textup{(iv)}, with symmetry parameter $\bar\nu\leq eS\lambda^2d$. It remains to verify
condition~\textup{(v)}, first for $g$ and then for the scaled $\widehat g$.

We first record a few Lewis-weight identities.
\begin{fact}[Lewis-weight derivatives
{\cite[Definition~20 and Lemma~24, 31, 47]{LS19solving}}]
\label{fact:DWh}
For $x\in\interior\K$, set
\begin{align*}
    P_x
    &\defn
    W_x^\beta A_xg(x)^{-1}A_x\tp W_x^\beta\,,\\
    \bar\Lambda_x
    &\defn
    I-W_x^{-1/2}P_x^{\circ2}W_x^{-1/2}\,,\\
    N_x
    &\defn
    2\bar\Lambda_x(I-\alpha\bar\Lambda_x)^{-1}\,.
\end{align*}
Then, $P_x$ is the orthogonal projector onto $\cols(W_x^\beta A_x)$, $N_x$ is symmetric, and the maps $x\mapsto P_x$ and $x\mapsto N_x$ are smooth. Moreover,
for every $h\in\Rd$,
\begin{equation}
\begin{aligned}
    \diag P_x&=w_x,
    &0\preceq\bar\Lambda_x&\preceq I,
    &0\preceq N_x&\preceq pI,\\
    \Dd\log w_x[h]
    &=-W_x^{-1/2}N_xW_x^{1/2}A_xh\,.
\end{aligned}
\label{eq:lewisBasic-PWI}
\end{equation}
\end{fact}

Fix $x\in\interior\K$. Conditions~\textup{(B1)}--\textup{(B3)} are
invariant under invertible affine changes of coordinates, so we may
assume $x=0$ and $g(0)=I_d$. Let $h\sim\Normal(0,I_d)$ and $x_t=th$.
Set $\xi_t\defn g(x_t)^{1/2}h$ and
$F(t)\defn h\tp g(x_t)h$.
Along the path, write $A_t=A_{x_t}$, $w_t=w_{x_t}$,
$W_t=W_{x_t}$, $P_t\defn P_{x_t}$, $N_t=N_{x_t}$, and
$\ell_t\defn\log w_t$, and set
\begin{align*}
    \iota_t&\defn A_th\,,
    &v_t&\defn W_t^{1/2}\iota_t = W_t^{1/2}A_th\,,
    &q_t&\defn W_t^{-1/2}\,(W_t^\beta \iota_t)^{\circ2}\,.
\end{align*}

\begin{fact}[Fourth-order identity for the LS metric {\cite[\S3.1]{Kook26Dikin}}] \label{fact:LS-fourth-order-path}
For every $t$ such that $x_t\in\interior\K$,
\begin{align}
F''(t) &= \Dd^2g(x_t)[h,h][h,h] = \cQ_{x_t}^{g}[\xi_t^{\otimes4}]\nonumber\\
    &=
    6\sum_i w_{t,i}^\alpha \iota_{t,i}^4
    +2\alpha\,\inner{N_tv_t,q_t\circ \iota_t}
    -2\beta\,\bpar{
        (q_t')\tp N_tv_t
        +q_t\tp N_tv_t'
        +q_t\tp N_t'v_t
    }\,.\label{eq:LS-fourth-order-path-decomposition}
\end{align}
\end{fact}

\paragraph{The bottleneck.}
\cite[Lemma~3.3 and~3.4]{Kook26Dikin} already control the four terms (at the required scale) in decomposition~\eqref{eq:LS-fourth-order-path-decomposition} that do not contain $N_t'$. For the remaining term, the direct estimate $\norm{N_t'}_{\op}=\Otilde(\sqrt d)$ in \cite{Kook26Dikin} gives only $\abs{q_t\tp N_t'v_t}=\Otilde(d^{3/2})$, whereas condition~\textup{(B3)} requires an $O(d)$-bound times the certificate introduced below. 

Motivated by this, we introduce the following certificate.
For $y\in\interior\K$ and $h\in\Rd$, define
\[
    D_y(h)
    \defn
    \Diag(\beta\,\Dd\log w_y[h]-A_yh)\,.
\]
The identity
$\Dd(W_y^\beta A_y)[h]=D_y(h)\,W_y^\beta A_y$ expresses this derivative as row scaling by $D_y(h)$.
Since $P_y$ is the orthogonal projector onto $\cols (W_y^\beta A_y)$ and determines $N_y$ through Fact~\ref{fact:DWh}, define
\[
    \omega_{\mathrm{LS}}(y,h)
    \defn
    1+\norm{D_y(h)\,P_y}_{\op}\,.
\]

\paragraph{Why this certificate?}
Along the path, let $D_t\defn D_{x_t}(h)$ and $\omega_t\defn\omega_{\mathrm{LS}}(x_t,h)$.
Differentiating its orthogonal projector $P_t = \sigma(W_t^\beta A_t)$ gives
\[
    P_t'=(I-P_t)\,D_tP_t+P_tD_t\,(I-P_t)\,.
\]
Hence, $\norm{P_t'}_{\op}\leq2\,\norm{D_tP_t}_{\op}=2\,(\omega_t-1)$. Passing from $P_t'$ to $N_t'$ is the key nontrivial step, and Lemma~\ref{lem:operator-row-map-estimate} proves $\norm{N_t'}_{\op}\leq6p^2\,(\omega_t-1)$.
Combined with the remaining path estimates, this bound controls $q_t\tp N_t'v_t$ and hence $F''(t)$ at the scale required for~\textup{(B3)}. The same bound controls the contribution of $N_t'$ to $D_t'$ and yields a differential inequality for $D_tP_t$, which gives~\textup{(B2)} upon integration. Thus, $\omega_t$ can address a previous bottleneck without the extra $\sqrt d$ loss of the direct estimate. Lemma~\ref{lem:spectral-LS-short-bootstrap} makes these two consequences precise.

In Lemma~\ref{lem:spectral-LS-background-event}, we streamline the argument of Lemma~3.3 in~\cite{Kook26Dikin}, tracking the dependence on $p,m,\delta$ explicitly. It verifies \textup{(B1)} and records the required path estimates, with its proof deferred to \S\ref{app:LS-pointwise-path-estimates}.
The proofs of Lemmas~\ref{lem:operator-row-map-estimate} and~\ref{lem:spectral-LS-short-bootstrap} are deferred to the end of the subsection.

Below, set $\varkappa_\delta \defn p^2\,(1+\log\frac{8m}{\delta})^{1/2}$ for $\delta\in(0,1/2)$.

\begin{lemma}[Lewis-weight path estimates]
\label{lem:spectral-LS-background-event}
Fix $\delta\in(0,1/2)$ and write $\varkappa=\varkappa_\delta$. There are
universal constants $c,C>0$ and an event $\cE_\delta$ of probability at
least $1-\delta$ such that, for
$0\leq t\leq c\varkappa^{-4}/\sqrt d$,
\begin{align*}
    x_t&\in\interior\K\,,
    &\omega_0&\leq C\varkappa\,,\\
    \norm{h}_{g(x_t)}&\leq C\varkappa\sqrt d\,,
    &\norm{\iota_t}_\infty&\leq C\varkappa\,,\\
    \norm{\ell_t'}_\infty+\norm{D_t}_{\op}
    &\leq C\varkappa^2\sqrt d\,,
    &\norm{v_t}&\leq C\varkappa\sqrt d\,,\\
    \norm{N_tv_t}+\norm{q_t}
    &\leq C\varkappa^2\sqrt d\,,
    &\norm{v_t'}&\leq C\varkappa^3\sqrt d\,,\\
    \norm{q_t'}&\leq C\varkappa^4\sqrt d\,.
\end{align*}
Moreover, $t\mapsto\omega_t$ is continuous on this segment.
\end{lemma}

\begin{lemma}[Derivative bound for $N_t$]\label{lem:operator-row-map-estimate}
For every $t$ for which $x_t\in\interior\K$, $\norm{N_t'}_{\op} \leq 6p^2(\omega_t-1)$.
\end{lemma}

\begin{lemma}[Bootstrap control and fourth-order dominance]
\label{lem:spectral-LS-short-bootstrap}
Fix $\delta\in(0,1/2)$ and write $\varkappa=\varkappa_\delta$.
There exists $C \asymp 1$ such that, on the event $\cE_\delta$ in
Lemma~\ref{lem:spectral-LS-background-event}, for every
$0\leq t\leq c\varkappa^{-4}/\sqrt d$,
\begin{align*}
    \omega_t
    &\leq
    \omega_0
    +C\varkappa^4\sqrt d\int_0^t\omega_s\,\D s\,,
    &\abs{\cQ_{x_t}^{g}[\xi_t^{\otimes4}]}
    &\leq C\varkappa^6d\,\omega_t\,.
\end{align*}
\end{lemma}

\begin{proof}[Proof of Corollary~\ref{cor:spectral-LS-mixing}]
Fix $\delta\in(0,1/2)$. Since the base point $x$ is arbitrary,
Lemma~\ref{lem:spectral-LS-background-event} verifies
condition~\textup{(B1)}, including continuity, and
Lemma~\ref{lem:spectral-LS-short-bootstrap} verifies
\textup{(B2)}--\textup{(B3)} on the same event
$\cE_\delta$. 

We now check the effects of $\Otilde(1)$-scaling of $g_{\textrm{LS}}$ on the problem parameters. Intuitively, this should not affect the final complexity, but we provide the scaling details for completeness.
After adjusting universal constants, baseline parameters for $(g,\omega_{\mathrm{LS}})$ are
\begin{align*}
    r_\delta&=c\varkappa_\delta^{-4}\,,
    &\eta_\delta&=C\varkappa_\delta^4\,,
    &b_\delta&=C\varkappa_\delta^6\,,
    &c_\delta&=C\varkappa_\delta\,.
\end{align*}
Set $S\defn C_0(1+\log m)^{20}$, $L\defn S\lambda^2$, and $\widehat g\defn Lg$, where $C_0$ is a sufficiently large universal constant.
Fix once and for all $\delta_\star=1/32$, set
$\varkappa_\star\defn\varkappa_{\delta_\star}$, and define
\[
    \widehat\omega_{\mathrm{LS}}(x,\widehat h)
    \defn
    1+\varkappa_\star^{-1}\,\bpar{
        \omega_{\mathrm{LS}}(x,\sqrt L\,\widehat h)-1
    }\,.
\]
As $\delta_\star$ is fixed, $\widehat\omega_{\mathrm{LS}}$ does not depend on the running value of $\delta$.
The corresponding parameters are
\begin{equation}
    \widehat r_\delta=\sqrt L\,r_\delta\,,
    \qquad
    \widehat\eta_\delta=\frac{\eta_\delta}{\sqrt L}\,,
    \qquad
    \widehat b_\delta=\frac{\varkappa_\star b_\delta}{L}\,,
    \qquad
    \widehat c_\delta=1+\frac{c_\delta-1}{\varkappa_\star}\,.
\label{eq:fourth-order-scaling-profile}
\end{equation}
To verify this scaling rule, fix $x\in\interior\K$ and
$\xi\sim\Normal(0,I_d)$, and set
$h=g(x)^{-1/2}\xi$ and
$\widehat h=\widehat g(x)^{-1/2}\xi=L^{-1/2}h$. Write
\begin{align*}
    \omega_t&\defn\omega_{\mathrm{LS}}(x+th,h),
    &\xi_t&\defn g(x+th)^{1/2}h,\\
    \widehat\omega_t
    &\defn\widehat\omega_{\mathrm{LS}}(x+t\widehat h,\widehat h),
    &\widehat\xi_t
    &\defn\widehat g(x+t\widehat h)^{1/2}\widehat h.
\end{align*}
Since $x+t\widehat h=x+(t/\sqrt L)h$,
\begin{align*}
    \widehat\omega_t
    &=1+\varkappa_\star^{-1}(\omega_{t/\sqrt L}-1),
    &\cQ_{x+t\widehat h}^{\widehat g}
    [\widehat\xi_t^{\otimes4}]
    &=L^{-1}\cQ_{x+(t/\sqrt L)h}^{g}
      [\xi_{t/\sqrt L}^{\otimes4}].
\end{align*}
The same event $\cE_\delta$ gives feasibility and continuity for
$0\leq t\leq\sqrt L\,r_\delta/\sqrt d$, and
$\widehat\omega_0\leq
1+(c_\delta-1)/\varkappa_\star$; hence condition~\textup{(B1)} holds.
On this interval,
$\omega_{s/\sqrt L}\leq\varkappa_\star\widehat\omega_s$.
Applying the baseline conditions~\textup{(B2)}--\textup{(B3)}, using
the preceding identities for $\widehat\omega_t$ and
$\cQ^{\widehat g}$, and changing variables in
condition~\textup{(B2)}, yields
\begin{align*}
    \widehat\omega_t
    &\leq
    \widehat\omega_0
    +\frac{\eta_\delta}{\varkappa_\star\sqrt L}\sqrt d
      \int_0^t\omega_{s/\sqrt L}\,\D s
    \leq
    \widehat\omega_0
    +\frac{\eta_\delta}{\sqrt L}\sqrt d
      \int_0^t\widehat\omega_s\,\D s,\\
    \abs{\cQ_{x+t\widehat h}^{\widehat g}
      [\widehat\xi_t^{\otimes4}]}
    &\leq
    \frac{b_\delta}{L}d\,\omega_{t/\sqrt L}
    \leq
    \frac{\varkappa_\star b_\delta}{L}d\,\widehat\omega_t.
\end{align*}
The displayed inequalities verify conditions~\textup{(B2)} and
\textup{(B3)}, thereby establishing the scaling rule~\eqref{eq:fourth-order-scaling-profile}.

By the definition of $\varkappa_\delta$ and the choice
$\delta_\star=1/32$,
\begin{equation}
    \frac{\varkappa_\delta}{\varkappa_\star}
    \leq
    \sqrt{1+\log\frac1\delta}.
\label{eq:concentration-profile-comparison}
\end{equation}
Moreover,
$\varkappa_\star^8\lesssim(1+\log m)^{20}$, so increasing $C_0$ if
necessary ensures $S\geq C\varkappa_\star^8$.
Since $\lambda\geq1$, we have $L\geq S\geq C\varkappa_\star^8$.
Substituting $r_\delta=c\varkappa_\delta^{-4}$, $\eta_\delta=C\varkappa_\delta^4$, $b_\delta=C\varkappa_\delta^6$, and $c_\delta=C\varkappa_\delta$ into~\eqref{eq:fourth-order-scaling-profile}, the comparison in~\eqref{eq:concentration-profile-comparison} shows that $\widehat r_\delta$ is bounded below by a positive constant depending only on $\delta$, while $\widehat\eta_\delta$, $\widehat b_\delta$, and $\widehat c_\delta$ are bounded above by constants depending only on $\delta$.
Decreasing $\widehat r_\delta$ to its lower bound and increasing the other three parameters to their upper bounds shows that $\widehat g$ has fourth-order bootstrap control in the sense of Definition~\ref{def:fourth-order-bootstrap-control}.
Thus, $\widehat g$ satisfies all conditions of Theorem~\ref{theorem:generic-pointwise-mixing}, with symmetry parameter $\bar\nu\leq eS\lambda^2d$. The theorem therefore yields
\[
    \chi^2(\nu_0\widetilde{\cT}^{\,N}\mmid\pi_\theta)
    \leq
    \exp\bpar{-\frac{cN}{S\lambda^2d^2}}\,
    \chi^2(\nu_0\mmid\pi_\theta)\,.
\]
Finally, $\lambda=1+\log m$, so $S\lambda^2=C_0(1+\log m)^{22}$ as claimed.
\end{proof}

\paragraph{Deferred proofs of Lemmas~\ref{lem:operator-row-map-estimate} and~\ref{lem:spectral-LS-short-bootstrap}.}

\begin{proof}[Proof of Lemma~\ref{lem:operator-row-map-estimate}]
Fix $t$. Along the path, write $P_s=P_{x_s}$, $\bar\Lambda_s=\bar\Lambda_{x_s}$, and $N_s=N_{x_s}$. At $s=t$, suppress subscripts and write $P=P_t$, $D=D_t$, $W=W_t$, $w=w_t$, $\bar\Lambda=\bar\Lambda_t$, $N=N_t$, and $\omega=\omega_t$. Then, $\norm{DP}_{\op}=\omega-1$.

We first reduce the claim to a bound on $\bar\Lambda'$.
Set $R_s=(I-\alpha\bar\Lambda_s)^{-1}$ and write $R=R_t$. Since $0\preceq\bar\Lambda\preceq I$, we have $\norm{R}_{\op}\leq p/2$.
Moreover, $R'=\alpha R\bar\Lambda'R$ and $I+\alpha\bar\Lambda R=R$. 
Differentiating $N=2\bar\Lambda R$ therefore gives $N'=2R\bar\Lambda'R$ and $\norm{N'}_{\op} \leq\frac{p^2}{2}\,\norm{\bar\Lambda'}_{\op}$.
Thus, it suffices to prove $\norm{\bar\Lambda'}_{\op}\leq12(\omega-1)$.

To control $\bar\Lambda'$, define the following linear map
$\Psi_s:\R^{m\times m}\to\R^m$, where the domain is equipped with the Frobenius norm. A direct calculation using $P_s^2=P_s$ also gives
\[
    \Psi_sX
    \defn
    W_s^{-1/2}\diag(P_sXP_s)\,,\qquad
    \Psi_s\Psi_s\tp
    =W_s^{-1/2}P_s^{\circ2}W_s^{-1/2}\,,\qquad
    \bar\Lambda_s=I-\Psi_s\Psi_s\tp\,.
\]
At $s=t$, write $\Psi=\Psi_t$ and $\Psi'=\left.\frac{\D}{\D s}\Psi_s\right|_{s=t}$.
Since $0\preceq\bar\Lambda\preceq I$, we have $\norm{\Psi}_{\op}\leq1$, and hence $\norm{\bar\Lambda'}_{\op} \leq 2\,\norm{\Psi'}_{\op}$.
It remains to prove $\norm{\Psi'}_{\op} \leq 6\,(\omega-1)$.

Let us define $\Gamma \defn \frac12\,W^{-1}W' =\Diag(\frac{w_i'}{2w_i})$. 
Differentiating $\Psi_s$ at $s=t$ gives
\[
    \Psi'X
    =W^{-1/2}\diag(P'XP+PXP'-\Gamma PXP)\,.
\]
For $B_s\defn W_s^\beta A_s$, differentiating $P_s=B_s(B_s^\T B_s)^{-1}B_s^\T$ using $B_s'=D_sB_s$ gives
\begin{equation}\label{eq:P-diff}
   P_s' = (I-P_s)D_sP_s+P_sD_s(I-P_s)\qquad
   \norm{P'}_{\op} \leq 2\,\norm{DP}_{\op} = 2\,(\omega-1)\,.
\end{equation}

Using $w=\diag P$ with the formula for $P'$, observe that 
\[
\frac{w_i'}{2w_i} = D_{ii}-\frac{(PDP)_{ii}}{w_i}\,.
\]
Since $w_i = P_{ii} = \norm{Pe_i}^2$, the quotient $(PDP)_{ii}/w_i$ is bounded by $\norm{PDP}_\op$. Consequently,
\[
    \norm{\Gamma-D}_{\op}
    \leq\norm{PDP}_{\op}
    \leq\omega-1\,,
    \qquad
    \norm{\Gamma P}_{\op}
    \leq
    \norm{DP}_{\op}+\norm{(\Gamma-D)P}_{\op}
    \leq 2\,(\omega-1)\,.
\]

For any $M,X\in\R^{m\times m}$, rowwise Cauchy--Schwarz with $P_{ii}= w_i$ gives
\begin{align*}
    \norm{W^{-1/2}\diag(MXP)}
    &\leq\norm{MX}_{\F}\,,
    &\norm{W^{-1/2}\diag(PXM\tp)}
    &\leq\norm{XM\tp}_{\F}\,.
\end{align*}
Applying these inequalities with $M=P'$ and $M=\Gamma P$ yields
\[
    \norm{\Psi'X}
    \leq
    (2\,\norm{P'}_{\op}+\norm{\Gamma P}_{\op})\,
    \norm{X}_{\F}
    \leq 6\,(\omega-1)\norm{X}_{\F}\,.
\]
Combining the bounds on $\bar\Lambda'$ and $\Psi'$ proves
\[
    \norm{N_t'}_{\op}
    \leq
    \frac{p^2}{2}\,\norm{\bar\Lambda'}_{\op}
    \leq
    p^2\,\norm{\Psi'}_{\op}
    \leq
    6p^2\,(\omega_t-1)\,,
\]
which completes the proof.
\end{proof}

\begin{proof}[Proof of Lemma~\ref{lem:spectral-LS-short-bootstrap}]
Work on $\cE_\delta$ for $0\leq t\leq c\varkappa^{-4}/\sqrt d$.
Lemma~\ref{lem:operator-row-map-estimate}, together with
$p^2\leq\varkappa$ and $\omega_t = 1 + \norm{D_tP_t}_\op\geq1$, gives $\norm{N_t'}_{\op} \leq C\varkappa\,(\omega_t-1) \leq C\varkappa\omega_t$ throughout the interval.

We first prove bootstrap control by bounding $\norm{D_t'P_t}_\op$ and $\norm{D_tP_t'}_\op$.
Write $z_t=\beta\ell_t'-\iota_t$, so that $D_t=\Diag z_t$. 
Fact~\ref{fact:DWh} and direct differentiation give $\iota_t' =-\iota_t^{\circ2}$ and $\ell_t' =-W_t^{-1/2}N_tv_t$. Hence, 
\[
    z_t'
    =-\frac\beta2\,(\ell_t')^{\circ2}
     -\beta W_t^{-1/2}N_t'v_t
     -\beta W_t^{-1/2}N_tv_t'
     +\iota_t^{\circ2}\,.
\]
For $y,r\in\R^m$, the identity $\sum_j(P_t)_{ij}^2=(P_t)_{ii}=w_{t,i}$ and the inequality $\norm{\cdot}_{\op}\leq\norm{\cdot}_{\F}$ imply
\[
    \norm{\Diag(W_t^{-1/2}y)\,P_t}_{\op}
    \leq\norm{y}\,,\qquad
    \norm{\Diag(r^{\circ2})\,P_t}_{\op}
    \leq
    \norm{r}_\infty\norm{(\Diag r)\,P_t}_{\op}\,.
\]
We use the first inequality to bound the two middle terms in $z_t'$ and the second one to handle the rest.
Using $(\Diag\ell_t')\,P_t =\beta^{-1}\parenth{D_tP_t+(\Diag \iota_t)P_t}$ (which follows from $\ell_t' = \beta^{-1}\,(z_t + \iota_t)$), together with Lemmas~\ref{lem:spectral-LS-background-event} and~\ref{lem:operator-row-map-estimate}, we obtain
\[
    \norm{D_t'P_t}_{\op} \leq C\varkappa^4\sqrt d\,\omega_t +C\varkappa\sqrt d\,\norm{N_t'}_{\op}
    \leq C\varkappa^4\sqrt d\,\omega_t\,.
\]
Also, by \eqref{eq:P-diff} and Lemma~\ref{lem:spectral-LS-background-event}, $\norm{D_tP_t'}_{\op} \leq C\varkappa^2\sqrt d\,(\omega_t-1) \leq C\varkappa^4\sqrt d\,\omega_t$.
Hence, $\norm{(D_tP_t)'}_{\op} \leq C\varkappa^4\sqrt d\,\omega_t$.
By integrating this derivative bound,
\[
    \omega_t
    \leq
    \omega_0+\norm{D_tP_t-D_0P_0}_{\op}
    \leq
    \omega_0
    +C\varkappa^4\sqrt d\int_0^t\omega_s\,\D s\,.
\]

For fourth-order dominance, Lemmas~\ref{lem:spectral-LS-background-event} and~\ref{lem:operator-row-map-estimate} give
\begin{align*}
    \sum_iw_{t,i}^\alpha \iota_{t,i}^4
    &\leq\norm{\iota_t}_\infty^2\norm{h}_{g(x_t)}^2
    \leq C\varkappa^4d\,,\\
    \abs{\inner{N_tv_t,q_t\circ \iota_t}}
    &\leq C\varkappa^5d\,,\\
    \abs{(q_t')\tp N_tv_t}
    +\abs{q_t\tp N_tv_t'}
    &\leq C\varkappa^6d\,,\\
    \abs{q_t\tp N_t'v_t}
    &\leq C\varkappa^3d\,\norm{N_t'}_{\op}
    \leq C\varkappa^4d\omega_t\,.
\end{align*}
Therefore, by Fact~\ref{fact:LS-fourth-order-path}, $\abs{\cQ_{x_t}^{g}[\xi_t^{\otimes4}]} =\abs{F''(t)} \leq C\varkappa^6d\omega_t$.
\end{proof}

\begin{acknowledgements}
The authors thank the FIM--Institute for Mathematical Research at ETH Z\"urich for supporting the conference ``Scalable MCMC Sampling'', where this project was initiated.
The authors also thank Santosh Vempala and Minhui Jiang for discussions at an early stage of the project.
YK is supported in part by NSF Awards CCF-2504995, CCF-2236669, and CCF-2504994. YK thanks Victor Reis for sharing a preliminary proof of Corollary~4 found by ChatGPT 5.6 Pro based on~\cite{Kook26Dikin}.
\end{acknowledgements}

\bibliographystyle{alpha}
\bibliography{ref}

@article{KV26HAR,
	author = {Kook, Yunbum and Vempala, Santosh S.},
	journal = {arXiv preprint arXiv:2608.16878},
	title = {Spectral gaps of {H}it-and-{R}un and {C}oordinate {H}it-and-{R}un},
	year = {2026}}

@inproceedings{LLV20strong,
	author = {Laddha, Aditi and Lee, Yin Tat and Vempala, Santosh S.},
	booktitle = {{S}ymposium on {T}heory of {C}omputing},
	collection = {STOC '20},
	doi = {10.1145/3357713.3384272},
	month = jun,
	pages = {1212--1222},
	publisher = {ACM},
	title = {Strong self-concordance and sampling},
	url = {http://dx.doi.org/10.1145/3357713.3384272},
	year = {2020}}

@article{LS93random,
	author = {Lov\'{a}sz, L\'{a}szl\'{o} and Simonovits, Mikl\'{o}s},
	doi = {10.1002/rsa.3240040402},
	fjournal = {Random Structures \& Algorithms},
	issn = {1042-9832},
	journal = {Random Structures \& Algorithms},
	mrclass = {90C27 (52B55 90C15)},
	mrnumber = {1238906},
	mrreviewer = {Gerard Sierksma},
	number = {4},
	pages = {359--412},
	title = {Random walks in a convex body and an improved volume algorithm},
	url = {https://doi.org/10.1002/rsa.3240040402},
	volume = {4},
	year = {1993}}

@article{lawler1988bounds,
	author = {Lawler, Gregory F. and Sokal, Alan D.},
	doi = {10.1090/S0002-9947-1988-0930082-9},
	journal = {Transactions of the American Mathematical Society},
	number = {2},
	pages = {557--580},
	title = {Bounds on the {$L^2$} Spectrum for {Markov} Chains and {Markov} Processes: A Generalization of {Cheeger}'s Inequality},
	volume = {309},
	year = {1988}}

@article{Vaidya96convex,
	author = {Vaidya, Pravin M.},
	doi = {10.1016/0025-5610(92)00021-S},
	fjournal = {Mathematical Programming},
	issn = {0025-5610,1436-4646},
	journal = {Mathematical Programming},
	mrclass = {90C25 (65K05)},
	mrnumber = {1393123},
	mrreviewer = {M.\ Teboulle},
	number = {3, Ser. A},
	pages = {291--341},
	title = {A new algorithm for minimizing convex functions over convex sets},
	url = {https://doi.org/10.1016/0025-5610(92)00021-S},
	volume = {73},
	year = {1996}}

@book{janson_gaussian_1997,
	author = {Janson, Svante},
	copyright = {https://www.cambridge.org/core/terms},
	doi = {10.1017/CBO9780511526169},
	edition = {1},
	isbn = {978-0-521-56128-0 978-0-521-05720-2 978-0-511-52616-9},
	month = jun,
	publisher = {Cambridge University Press},
	title = {Gaussian {Hilbert} {Spaces}},
	url = {https://www.cambridge.org/core/product/identifier/9780511526169/type/book},
	urldate = {2024-11-25},
	year = {1997}}

@inproceedings{CP15lprow,
	address = {New York, NY, USA},
	author = {Cohen, Michael B. and Peng, Richard},
	booktitle = {{S}ymposium on {T}heory of {C}omputing},
	pages = {183--192},
	publisher = {ACM},
	series = {STOC '15},
	title = {$L_p$ row sampling by {L}ewis weights},
	year = {2015}}

@article{SV16mixing,
	author = {Sachdeva, Sushant and Vishnoi, Nisheeth K.},
	doi = {10.1016/j.orl.2016.07.005},
	fjournal = {Operations Research Letters},
	issn = {0167-6377,1872-7468},
	journal = {Operations Research Letters},
	mrclass = {60D05 (52B55 60G50 90C05)},
	mrnumber = {3546883},
	number = {5},
	pages = {630--634},
	title = {The mixing time of the {D}ikin walk in a polytope---a simple proof},
	url = {https://doi.org/10.1016/j.orl.2016.07.005},
	volume = {44},
	year = {2016}}

@article{CV18Gaussian,
	author = {Cousins, Ben and Vempala, Santosh S.},
	doi = {10.1137/15M1054250},
	fjournal = {SIAM Journal on Computing},
	issn = {0097-5397},
	journal = {SIAM Journal on Computing},
	mrclass = {68W20 (52A38 60D05 65C40)},
	mrnumber = {3818340},
	number = {3},
	pages = {1237--1273},
	publisher = {Society for Industrial and Applied Mathematics (SIAM)},
	title = {Gaussian cooling and {$O^*(n^3)$} algorithms for volume and {G}aussian volume},
	url = {https://doi.org/10.1137/15M1054250},
	volume = {47},
	year = {2018}}

@article{KLS97random,
	author = {Kannan, Ravi and Lov\'{a}sz, L\'{a}szl\'{o} and Simonovits, Mikl\'{o}s},
	doi = {10.1002/(SICI)1098-2418(199708)11:1<1::AID-RSA1>3.0.CO;2-X},
	fjournal = {Random Structures \& Algorithms},
	issn = {1042-9832},
	journal = {Random Structures \& Algorithms},
	mrclass = {68Q25 (52A20 52A38 60J15 60J20)},
	mrnumber = {1608200},
	mrreviewer = {Mark R. Jerrum},
	number = {1},
	pages = {1--50},
	title = {Random walks and an {$O^*(n^5)$} volume algorithm for convex bodies},
	url = {https://doi.org/10.1002/(SICI)1098-2418(199708)11:1<1::AID-RSA1>3.0.CO;2-X},
	volume = {11},
	year = {1997}}

@article{Lovasz99hit,
	author = {Lov\'{a}sz, L\'{a}szl\'{o}},
	doi = {10.1007/s101070050099},
	id = {Lov{\'a}sz1999},
	isbn = {1436-4646},
	journal = {Mathematical Programming},
	number = {3},
	pages = {443--461},
	title = {Hit-and-{R}un mixes fast},
	url = {https://doi.org/10.1007/s101070050099},
	volume = {86},
	year = {1999}}

@article{LV06simulated,
	author = {Lov\'{a}sz, L\'{a}szl\'{o} and Vempala, Santosh S.},
	doi = {10.1016/j.jcss.2005.08.004},
	fjournal = {Journal of Computer and System Sciences},
	issn = {0022-0000},
	journal = {Journal of Computer and System Sciences},
	mrclass = {68U05 (52A20 60C05 68W20)},
	mrnumber = {2205290},
	mrreviewer = {Carla Peri},
	number = {2},
	pages = {392--417},
	title = {Simulated annealing in convex bodies and an {$O^*(n^4)$} volume algorithm},
	url = {https://doi.org/10.1016/j.jcss.2005.08.004},
	volume = {72},
	year = {2006}}

@article{LV06hit,
	author = {Lov\'{a}sz, L\'{a}szl\'{o} and Vempala, Santosh S.},
	doi = {10.1137/s009753970544727x},
	issn = {1095-7111},
	journal = {SIAM Journal on Computing},
	number = {4},
	pages = {985--1005},
	publisher = {Society for Industrial and Applied Mathematics (SIAM)},
	title = {Hit-and-{R}un from a corner},
	url = {http://dx.doi.org/10.1137/S009753970544727X},
	volume = {35},
	year = {2006}}

@article{LV07geometry,
	author = {Lov\'{a}sz, L\'{a}szl\'{o} and Vempala, Santosh S.},
	doi = {10.1002/rsa.20135},
	fjournal = {Random Structures \& Algorithms},
	issn = {1042-9832},
	journal = {Random Structures \& Algorithms},
	mrclass = {94A20 (60C05 60J10 60J22 62D05 65C10)},
	mrnumber = {2309621},
	mrreviewer = {Hsien-Kuei Hwang},
	number = {3},
	pages = {307--358},
	title = {The geometry of logconcave functions and sampling algorithms},
	url = {https://doi.org/10.1002/rsa.20135},
	volume = {30},
	year = {2007}}

@article{KN12random,
	author = {Kannan, Ravi and Narayanan, Hariharan},
	doi = {10.1287/moor.1110.0519},
	fjournal = {Mathematics of Operations Research},
	issn = {0364-765X,1526-5471},
	journal = {Mathematics of Operations Research},
	mrclass = {65K05 (90C05 90C55)},
	mrnumber = {2891144},
	number = {1},
	pages = {1--20},
	publisher = {Institute for Operations Research and the Management Sciences (INFORMS)},
	title = {Random walks on polytopes and an affine interior point method for linear programming},
	url = {https://doi.org/10.1287/moor.1110.0519},
	volume = {37},
	year = {2012}}

@article{CDW18MCMC,
	author = {Chen, Yuansi and Dwivedi, Raaz and Wainwright, Martin J. and Yu, Bin},
	fjournal = {Journal of Machine Learning Research (JMLR)},
	issn = {1532-4435,1533-7928},
	journal = {Journal of Machine Learning Research},
	mrclass = {60J22 (60D05 60J10 65C05)},
	mrnumber = {3874163},
	number = {55},
	pages = {1--86},
	title = {Fast {MCMC} sampling algorithms on polytopes},
	volume = {19},
	year = {2018}}

@article{LS19solving,
	author = {Lee, Yin Tat and Sidford, Aaron},
	journal = {arXiv preprint arXiv:1910.08033},
	title = {Solving linear programs with $\sqrt{\text{rank}}$ linear system solves},
	year = {2019}}

@article{Klartag23log,
	author = {Klartag, Bo'az},
	fjournal = {Ars Inveniendi Analytica},
	journal = {Ars Inveniendi Analytica},
	mrclass = {52A40 (58J65)},
	mrnumber = {4603941},
	mrreviewer = {Ge Xiong},
	number = {4},
	pages = {1--17},
	title = {Logarithmic bounds for isoperimetry and slices of convex sets},
	url = {https://mathscinet.ams.org/mathscinet-getitem?mr=4603941},
	year = {2023}}

@article{JLLV26reducing,
	address = {New York, NY, USA},
	articleno = {8},
	author = {Jia, He and Laddha, Aditi and Lee, Yin Tat and Vempala, Santosh S.},
	doi = {10.1145/3795687},
	issn = {0004-5411},
	issue_date = {April 2026},
	journal = {Journal of the ACM},
	month = apr,
	number = {2},
	numpages = {21},
	publisher = {Association for Computing Machinery},
	title = {Reducing isotropy and volume to {KLS}: faster rounding and volume algorithms},
	url = {https://doi.org/10.1145/3795687},
	volume = {73},
	year = {2026}}

@article{CE25hitandrun,
	author = {Chen, Yuansi and Eldan, Ronen},
	date = {2026/04/01},
	doi = {10.1007/s00454-025-00808-4},
	id = {Chen2026},
	isbn = {1432-0444},
	journal = {Discrete \& Computational Geometry},
	number = {3},
	pages = {747--794},
	title = {Hit-and-{R}un mixing via localization schemes},
	url = {https://doi.org/10.1007/s00454-025-00808-4},
	volume = {75},
	year = {2026}}

@inproceedings{KV24ipm,
	author = {Kook, Yunbum and Vempala, Santosh S.},
	booktitle = {Conference on Learning Theory},
	pages = {3137--3240},
	publisher = {PMLR},
	title = {Gaussian cooling and {D}ikin walks: the interior-point method for logconcave sampling},
	url = {https://proceedings.mlr.press/v247/kook24b.html},
	volume = {247},
	year = {2024}}

@article{DFK91random,
	author = {Dyer, Martin and Frieze, Alan and Kannan, Ravi},
	doi = {10.1145/102782.102783},
	fjournal = {Journal of the ACM},
	issn = {0004-5411},
	journal = {Journal of the ACM},
	mrclass = {68U05 (52B55 68Q25)},
	mrnumber = {1095916},
	number = {1},
	pages = {1--17},
	title = {A random polynomial-time algorithm for approximating the volume of convex bodies},
	url = {https://doi.org/10.1145/102782.102783},
	volume = {38},
	year = {1991}}

@article{KVZ26INO,
	author = {Kook, Yunbum and Vempala, Santosh S. and Zhang, Matthew S.},
	doi = {https://doi.org/10.1002/rsa.70061},
	eprint = {https://onlinelibrary.wiley.com/doi/pdf/10.1002/rsa.70061},
	journal = {Random Structures \& Algorithms},
	number = {3},
	pages = {e70061},
	title = {In-and-{O}ut: algorithmic diffusion for sampling convex bodies},
	url = {https://onlinelibrary.wiley.com/doi/abs/10.1002/rsa.70061},
	volume = {68},
	year = {2026}}

@article{WSC22minimax,
	author = {Wu, Keru and Schmidler, Scott and Chen, Yuansi},
	fjournal = {Journal of Machine Learning Research (JMLR)},
	issn = {1532-4435,1533-7928},
	journal = {Journal of Machine Learning Research},
	mrclass = {62D05 (65C05 70Hxx)},
	mrnumber = {4577709},
	number = {270},
	pages = {1--63},
	title = {Minimax mixing time of the metropolis-adjusted {L}angevin algorithm for log-concave sampling},
	volume = {23},
	year = {2022}}

@article{CG23when,
  title = {When does {M}etropolized {H}amiltonian {M}onte {C}arlo provably outperform {M}etropolis-adjusted {L}angevin algorithm?},
  author={Chen, Yuansi and Gatmiry, Khashayar and Jiang, Minhui},
  journal={arXiv preprint arXiv:2304.04724},
  year={2023}
}

@inproceedings{JC25regularized,
	author = {Jiang, Minhui and Chen, Yuansi},
	booktitle = {Conference on Learning Theory},
	month = {30 Jun--04 Jul},
	pages = {3017--3078},
	publisher = {PMLR},
	series = {Proceedings of Machine Learning Research},
	title = {Regularized {D}ikin walks for sampling truncated logconcave measures, mixed isoperimetry and beyond worst-case analysis},
	url = {https://proceedings.mlr.press/v291/jiang25a.html},
	volume = {291},
	year = {2025}}

@article{Kook26Dikin,
	author = {Kook, Yunbum},
	journal = {arXiv preprint arXiv:2607.13943},
	title = {Beyond the $d^{2.5}$-mixing bound for {D}ikin walks on polytopes},
	year = {2026}}

@book{NN94ipm,
	author = {Nesterov, Yurii and Nemirovskii, Arkadii},
	doi = {10.1137/1.9781611970791},
	isbn = {9781611970791},
	month = jan,
	publisher = {Society for Industrial and Applied Mathematics (SIAM)},
	title = {Interior-point polynomial algorithms in convex programming},
	url = {http://dx.doi.org/10.1137/1.9781611970791},
	year = {1994}}

@inproceedings{Dikin1967iterative,
	author = {Dikin, Iliya Iosiphovich},
	booktitle = {Soviet Math. Dokl.},
	pages = {674--675},
	title = {Iterative solution of problems of linear and quadratic programming},
	volume = {8},
	year = {1967}}

@inproceedings{Karmarkar84polytime,
	author = {Karmarkar, Narendra},
	booktitle = {{S}ymposium on {T}heory of {C}omputing},
	collection = {STOC '84},
	doi = {10.1145/800057.808695},
	pages = {302--311},
	publisher = {ACM Press},
	series = {STOC '84},
	title = {A new polynomial-time algorithm for linear programming},
	url = {http://dx.doi.org/10.1145/800057.808695},
	year = {1984}}

@article{KV26unified,
	author = {Kook, Yunbum and Vempala, Santosh S.},
	journal = {arXiv preprint arXiv:2606.12694},
	title = {A unified complexity bound for logconcave sampling},
	year = {2026}}

@article{KV26zeroLC,
	author = {Kook, Yunbum and Santosh S. Vempala},
	journal = {arXiv preprint arXiv:2507.18021},
	title = {Zeroth-order logconcave sampling},
	year = {2025}}

@inproceedings{KV25sampling,
	author = {Kook, Yunbum and Vempala, Santosh S.},
	booktitle = {{S}ymposium on {T}heory of {C}omputing},
	doi = {10.1145/3717823.3718202},
	isbn = {979-8-4007-1510-5},
	mrclass = {68Q87},
	mrnumber = {4928485},
	pages = {924--932},
	publisher = {ACM},
	title = {Sampling and integration of logconcave functions by algorithmic diffusion},
	url = {https://doi.org/10.1145/3717823.3718202},
	year = {2025}}

@article{Narayanan16randomized,
	author = {Narayanan, Hariharan},
	doi = {10.1214/15-AAP1104},
	fjournal = {The Annals of Applied Probability},
	issn = {1050-5164,2168-8737},
	journal = {The Annals of Applied Probability},
	mrclass = {65C40 (53C15 60J20 90C51)},
	mrnumber = {3449327},
	number = {1},
	pages = {597--641},
	title = {Randomized interior point methods for sampling and optimization},
	url = {https://doi.org/10.1214/15-AAP1104},
	volume = {26},
	year = {2016}}

@article{Lovasz90compute,
	author = {Lov\'{a}sz, L\'{a}szl\'{o}},
	journal = {Jber. d. Dt. Math.-Verein, Jubil{\"a}umstagung},
	pages = {138--151},
	publisher = {DIMACS, Center for Discrete Mathematics and Theoretical Computer Science},
	title = {How to compute the volume?},
	year = {1990}}

\appendix
\section{Deferred preliminaries} \label{app:deferred_preliminaries}
\paragraph{Notation.}
For $k \in \naturalnum$, let $[k]\defn\{1,\ldots,k\}$, and write $a\vee b$ and $a\wedge b$ for $\max\{a,b\}$ and $\min\{a,b\}$. 
We use $\inner{u,v}=u\tp v$ for the Euclidean inner product of vectors and $\inner{U,V}=\Tr(U\tp V)$ for the Frobenius inner product of matrices.
The norm $\norm{u}$ is Euclidean, while $\norm{U}_{\op}$ and $\norm{U}_{\F}$ are the operator and Frobenius norms.  We write $I_k$ and $\boldsymbol1$ for the identity matrix and all-ones vector, with dimensions omitted when clear.
We write $\Normal(\mu,\Sigma)$ for a Gaussian distribution with mean $\mu$ and covariance $\Sigma$.

The spaces of symmetric, positive-semidefinite, and positive-definite
$k\times k$ matrices are $\mathbb S^k$, $\mathbb S_+^k$, and
$\mathbb S_{++}^k$.  We write $U\preceq V$ when $V-U$ is positive
semidefinite and $\cols U$ for the column space of $U$.  For a vector
$v$ and a square matrix $U$, $\Diag v$ is the diagonal matrix with
diagonal $v$, while $\diag U$ is the diagonal vector of $U$.  Vector
inequalities and scalar functions applied to vectors are interpreted
entrywise.

\paragraph{Markov chains, mixing, and warm starts.}
We use standard Markov-chain notation; see, for example, Jiang and
Chen~\cite[\S1.1]{JC25regularized}. We take $\interior\K$
as the state space.  Let $\mathsf T$ be a Markov transition kernel on
$\interior\K$, and write
$\mathsf T_x\defn\mathsf T(x,\cdot)$.  Its action on a probability
measure $\mu$ is
\[
    (\mu\mathsf T)(E)
    \defn
    \int_{\interior\K}\mathsf T_x(E)\,\mu(\D x)\,,
    \qquad
    E\in\mathcal B(\interior\K)\,.
\]
Thus, $\nu_0\mathsf T^N$ is the law after $N$ steps from an initial
distribution $\nu_0$, and a probability measure $\pi$ is stationary if
$\pi\mathsf T=\pi$.  For a kernel with stationary distribution $\pi$
and $\veps>0$, define
\begin{align}
    \tmix^{\tv}(\veps;\nu_0,\mathsf T)
    &\defn
    \inf\bbrace{
        N\in\mathbb Z_{\geq0}:
        \dtv(\nu_0\mathsf T^N,\pi)\leq\veps
    }\,,\label{eq:mixing-tv}\\
    \tmix^{\chi^2}(\veps;\nu_0,\mathsf T)
    &\defn
    \inf\bbrace{
        N\in\mathbb Z_{\geq0}:
        \chi^2(\nu_0\mathsf T^N\mmid\pi)\leq\veps
    }\,.\label{eq:mixing-chi}
\end{align}
The lazy version of $\mathsf T$
is $\widetilde{\mathsf T}\defn(I+\mathsf T)/2$, where $I$ is the
identity kernel.  For $M\geq1$, the initialization $\nu_0$ is $M$-warm
with respect to $\pi$ if $\nu_0(E)\leq M\pi(E)$ for every
$E\in\mathcal B(\interior\K)$.
Equivalently, $\nu_0\ll\pi$ and $\D\nu_0/\D\pi\leq M$ almost surely. Note that $1+\chi^2(\nu_0\mmid\pi) = \int_{\interior\K}(\frac{\D\nu_0}{\D\pi})^2\,\D\pi \leq M$.

\section{Deferred proofs for Sections~\ref{sec:generic-overlap} and~\ref{sec:spectral-analysis}}

\subsection{Proof of Lemma~\ref{lem:proposal_overlap}}\label{subsec:generic-proposal-overlap}

The proposal-overlap argument uses only strong self-concordance. We use
the following consequence of SSC~\cite[Lemma~1.2]{LLV20strong}.

\begin{fact}
    \label{fact:SSC}
    If $g$ is SSC on $\interior\K$, then, for any $x,y\in\interior\K$ with
    $\norm{x-y}_{g(x)}<1$,
    \[
        \norm{
            g(x)^{-1/2}\parenth{g(y)-g(x)}g(x)^{-1/2}
        }_{\F}
        \leq
        \frac{2\norm{x-y}_{g(x)}}
        {(1-\norm{x-y}_{g(x)})^2}\,.
    \]
\end{fact}

\begin{proof}[Proof of Lemma~\ref{lem:proposal_overlap}]
    Let $\delta\defn\norm{x-y}_{g(x)}$ and $M\defn g(x)^{-1/2}g(y)g(x)^{-1/2}$.
    The matrix $M$ is symmetric positive definite; let
    $\lambda_1,\ldots,\lambda_d$ be its eigenvalues.
    Since $\delta\leq\step/10\leq1/20$, Fact~\ref{fact:SSC} implies
    \[
        \norm{M-I_d}_{\F}
        \leq\frac{2\delta}{(1-\delta)^2}
        \leq\frac94\,\delta
        \leq\frac9{80}<\frac18\,.
    \]
    Hence, all eigenvalues of $M$ satisfy $\lambda_i\in[7/8,9/8]$.

    Using the Gaussian density~\eqref{eq:proposal-density} directly,
    \begin{align*}
        2\,\mathrm{KL}(\proposal_y \mmid \proposal_x)
        &=
        2\int p_y(z)\log\frac{p_y(z)}{p_x(z)}\,\D z\\
        &=
        \log\frac{\det g(y)}{\det g(x)}
        +\frac{1}{\step^2}
        \int \norm{z-x}_{g(x)}^2\, p_y(z)\,\D z
        -\frac{1}{\step^2}
        \int \norm{z-y}_{g(y)}^2\, p_y(z)\,\D z\,.
    \end{align*}
    The relevant Gaussian moment identities are
    \begin{align*}
        \int \norm{z-x}_{g(x)}^2\, p_y(z)\,\D z
        &=
        \step^2\Tr\parenth{g(x)g(y)^{-1}} + \norm{y-x}_{g(x)}^2\,,\\
        \int \norm{z-y}_{g(y)}^2\, p_y(z) \,\D z
        &=
        \step^2d\,.
    \end{align*}
    Substituting these identities, we obtain
    \begin{align*}
        2\,\mathrm{KL}(\proposal_y \mmid \proposal_x)
        &=
        \log\frac{\det g(y)}{\det g(x)}
        +\Tr\parenth{g(x)g(y)^{-1}}-d
        +\frac{\norm{y-x}_{g(x)}^2}{\step^2}\\
        &=
        \sum_{i=1}^d\log\lambda_i
        +\sum_{i=1}^d(\lambda_i^{-1}-1)
        +\frac{\delta^2}{\step^2}\\
        &=
        \sum_{i=1}^d
        (\lambda_i^{-1}-1+\log\lambda_i)
        +\frac{\delta^2}{\step^2}\,.
    \end{align*}
    Using  $\lambda^{-1}-1+\log\lambda\leq(\lambda-1)^2$
    on $[7/8,9/8]$ and $\delta/\step\leq1/10$, we conclude that
    \[
        \mathrm{KL}(\proposal_y \mmid \proposal_x)
        \leq
        \frac12\,\Bpar{\sum_{i=1}^d(\lambda_i-1)^2 +\frac{\delta^2}{\step^2}}
        <\frac12\,\bpar{\frac1{64}+\frac1{100}}
        <\frac1{64}\,.
    \]
    By Pinsker's inequality,
    \[
        \dtv(\proposal_x,\proposal_y)
        \leq
        \sqrt{\frac12\,\mathrm{KL}(\proposal_y\mmid\proposal_x)}
        <
        \frac{1}{8\sqrt2}
        <
        \frac14\,,
    \]
    which completes the proof.
\end{proof}

\subsection{Proof of the conductance bound from transition overlap}
\label{subsec:proof-transition-overlap-conductance}

The following is the standard local-norm comparison for a
self-concordant local metric \cite[Lemma~1.1]{LLV20strong}. For Hessian metrics,
see also Nesterov and Nemirovskii
\cite[Theorem~2.1.1]{NN94ipm}.
\begin{lemma}[Local-norm comparison under self-concordance]\label{lem:self-concordant-local-norm-comparison}
Let $g$ be a self-concordant local metric on $\interior\K$, and let
$x,y\in\interior\K$. If $h\defn y-x$
satisfies $\norm{h}_{g(x)}<1$, then
\[
    \norm{h}_{g(y)}
    \leq
    \frac{\norm{h}_{g(x)}}{1-\norm{h}_{g(x)}}\,.
\]
\end{lemma}

\begin{proof}
    The claim is immediate when $h=0$. Otherwise, set $x_t\defn x+th$ and
    $\phi(t)\defn\norm{h}_{g(x_t)}$. By self-concordance of $g$,
    \[
        \phi'(t)
        =
        \frac{\Dd g(x_t)[h][h,h]}{2\phi(t)}
        \leq
        \phi(t)^2\,.
    \]
    Hence, $(1/\phi)'(t)\geq-1$. Integrating over $[0,1]$, we obtain
    \[
        \frac1{\phi(1)}
        \geq
        \frac1{\phi(0)}-1\,,
    \]
    which is equivalent to the claimed inequality.
\end{proof}

\begin{proof}[Proof of Lemma~\ref{lem:generic-warm-start-mixing}]
    Set $a\defn(1-\delta)/2$. Fix a measurable set $A$ such that
    $s<\pi_\theta(A)\leq1/2$, and consider the partition
    \begin{align*}
        A_1
        &\defn
        \{x\in A\cap\Xi: \mathsf T_x(A^c)<a\}\,,
        \\
        A_2
        &\defn
        \{y\in A^c\cap\Xi: \mathsf T_y(A)<a\}\,,
        \\
        A_3
        &\defn
        \K\setminus(A_1\cup A_2)\,.
    \end{align*}
    For any $x\in A_1$ and $y\in A_2$, the definitions imply
    $\mathsf T_x(A)>1-a$ and $\mathsf T_y(A)<a$.
    Hence,
    \[
        \dtv(\mathsf T_x,\mathsf T_y)
        \geq
        \mathsf T_x(A)-\mathsf T_y(A)
        >
        1-2a
        =
        \delta\,.
    \]
    The transition-overlap hypothesis implies that
    $\norm{x-y}_{g(x)}>\Delta$. We claim that
    $\norm{x-y}_{g(y)}\geq\Delta/2$. Otherwise,
    Lemma~\ref{lem:self-concordant-local-norm-comparison}, applied with
    base point $y$, would force a contradiction:
    \[
        \norm{x-y}_{g(x)}
        \leq
        \frac{\norm{x-y}_{g(y)}}{1-\norm{x-y}_{g(y)}}
        <
        \frac{\Delta/2}{1-\Delta/2}
        \leq
        \Delta\,.
    \]
    If $A_1$ and $A_2$ are nonempty, the assumed cross-ratio comparison implies
    \[
        d_{\K}(A_1,A_2)
        \geq
        \frac{\Delta}{2\sqrt{\bar\nu}}
        \geq
        \vartheta\,.
    \]

    Suppose $\pi_\theta(A_1)\leq\pi_\theta(A)/2$. Since $s\geq8\zeta/\vartheta$ and $\vartheta\leq1$, we have $\vartheta s\geq8\zeta$ and $s\geq8\zeta$. Thus,
    \begin{align*}
        Q_{\mathsf T}(A,A^c)
        &\geq
        \int_{(A\cap\Xi)\setminus A_1}
        \mathsf T_x(A^c)\,\pi_\theta(\D x)
        \ge
        a\,\bpar{
            \pi_\theta(A)-\zeta-\pi_\theta(A_1)
        }\\
        &\geq
        a\,\bpar{
            \frac{\pi_\theta(A)}2-\zeta
        }
        \geq
        \frac {a\vartheta}2\,\bpar{\pi_\theta(A)-s}\,.
    \end{align*}
    If $\pi_\theta(A_2)\leq \frac{\pi_\theta(A^c)}2$, apply the same argument to
    $A^c$. Reversibility and $\pi_\theta(A^c)\geq\pi_\theta(A)$ prove the
    same bound.

    Now consider the case $\pi_\theta(A_1)>\frac{\pi_\theta(A)}2$ and $\pi_\theta(A_2)>\frac{\pi_\theta(A^c)}2$.
    Fact~\ref{fact:cross-ratio-isoperimetry} and $\pi_\theta(A^c)\geq1/2$ yield
    \[
        \pi_\theta(A_3)
        \geq
        d_{\K}(A_1,A_2)\, \pi_\theta(A_1)\,\pi_\theta(A_2)
        \geq
        \frac{\vartheta}{4}\, \pi_\theta(A)\,\pi_\theta(A^c)
        \geq
        \frac{\vartheta\pi_\theta(A)}8\,.
    \]
    Every $x\in A_3\cap\Xi\cap A$ satisfies $\mathsf T_x(A^c)\geq a$,
    and $x\in A_3\cap\Xi\cap A^c$ satisfies $\mathsf T_x(A)\geq a$.
    By reversibility,
    \[
      2Q_{\mathsf T}(A,A^c) = Q_{\mathsf T}(A,A^c)+Q_{\mathsf T}(A^c,A) \geq a\,\pi_\theta(A_3\cap\Xi)\,.
    \]
    Since $\pi_\theta(A_3\cap\Xi)\geq\pi_\theta(A_3)-\zeta$ and $\zeta\leq\vartheta s/8$, it follows that
    \[
        Q_{\mathsf T}(A,A^c)
        \geq
        \frac a2\,\bpar{\frac{\vartheta\pi_\theta(A)}8-\zeta} \geq \frac{a\vartheta}{16}\, (\pi_\theta(A)-s)\,.
    \]

    In all cases,
    $Q_{\mathsf T}(A,A^c)\geq
    a\vartheta\parenth{\pi_\theta(A)-s}/16$.
    Taking the infimum over $A$ and using $a=(1-\delta)/2$, we obtain
    $\Phi_s(\mathsf T)\geq(1-\delta)\vartheta/32$.
    Lazification halves the ergodic flow, proving the claim.
\end{proof}

\subsection{One-step coupling under SSC, LTSC, and ASC}\label{sec:one_step_coupling_under_ssc_ltsc_and_asc}

\begin{proof}[Proof of Lemma~\ref{lem:pointwise-Dikin-framework}]
Choose a sufficiently small universal $c_0\in(0,1/2]$ and set
$r=c_0r_0$. For
$u\in\interior\K$, write $z=u+\step h$, where
$h\sim\Normal(0,g(u)^{-1})$, and whenever $u+th\in\interior\K$, set
\[
    F(t)\defn h\tp g(u+th)h\,,
    \qquad
    \ell(t) \defn\log\det g(u+th)\,.
\]
When the proposal is feasible, the log proposal-density ratio satisfies
\[
    \log\frac{p_z(u)}{p_u(z)}
    =
    \frac12\,\bpar{\ell(\step)-\ell(0)}
    -
    \frac12\,\bpar{F(\step)-F(0)}\,.
\]
Except on an event of probability at most $\rho_\star$, ASC ensures
feasibility and $\abs{F(\step)-F(0)}\leq2\rho_\star$. By convexity, the entire proposal segment is feasible whenever the proposal is feasible. Writing $\xi=g(u)^{1/2}h$, direct differentiation gives the
identity for $\ell'(0)$, SSC gives the trace bound, and SSC together with
LTSC gives the lower bound on $\ell''(t)$:
\[
    \ell'(0)
    =
    \inner{\Tr_{12}\cJ_u^g,\xi}\,,
    \qquad
    \norm{\Tr_{12}\cJ_u^g}
    \leq
    2\sqrt d\,,
    \qquad
    \ell''(t)
    \geq
    -5F(t)\,.
\]
Gaussian concentration, the local-norm comparison in Lemma~\ref{lem:self-concordant-local-norm-comparison}, and Taylor's formula therefore show, after decreasing $c_0$ if necessary, that $\ell(\step)-\ell(0)\geq-C$ outside an additional event of probability at most $2\rho_\star$, where $C$ is universal. This proposal-density-ratio estimate is standard; see \cite[(G.3)]{KV24ipm} and \cite[(55)]{JC25regularized}.

Set $a_\star=C/2+\rho_\star$. Since $3\rho_\star<1/16$, for every
$u\in\interior\K$,
\[
    \proposal_u(\mathcal G_u^c)
    \leq
    \frac1{16}\,,
    \qquad
    \mathcal G_u
    \defn
    \bbrace{z\in\interior\K:
    \frac{p_z(u)}{p_u(z)}\geq e^{-a_\star}}\,.
\]
For $x=y$, the claim is immediate. For distinct $x,y\in\interior\K$
satisfying the local-norm condition in the lemma, we apply the
transition-overlap argument in the proof of
Lemma~\ref{lem:acceptance_rate_control}, with
$\Xi=\interior\K$, $a=a_\star$, $\tau=1/16$, and the bad endpoint set
$\mathcal G_u^c$ at each $u$. The preceding bound on
$\proposal_u(\mathcal G_u^c)$ verifies the required bad-set estimate,
while the definition of $\mathcal G_u$ gives the same feasibility and
proposal-density-ratio guarantee. Therefore,
\[
    1-\dtv(\cT_x,\cT_y)
    \geq
    e^{-a_\star}\,\bpar{\frac14-\frac2{16}}
    =
    \frac{e^{-a_\star}}8\,,
\]
and the claim holds with $\delta_0=1-e^{-a_\star}/8\in(0,1)$.
\end{proof}

\section{Lewis-weight path estimates}
\label{app:LS-pointwise-path-estimates}

To prove Lemma~\ref{lem:spectral-LS-background-event}, we adapt the proof of Lemma~3.3 in~\cite{Kook26Dikin}, which establishes the Gaussian estimates at the base point and along the path while tracking $p,m,\delta$ explicitly.

\paragraph{Lewis-weight derivative and closeness estimates.}
Recall from Fact~\ref{fact:DWh} that $P_x$ is an orthogonal projector
with $\diag P_x=w_x$, that $0\preceq N_x\preceq pI$, and that both
matrices depend smoothly on $x$. Thus $0<w_{x,i}\leq1$ and
$\sum_iw_{x,i}=d$.

For the path $x_t=x+th$, write $A_t=A_{x_t}$, $W_t=W_{x_t}$,
$w_t=w_{x_t}$, and $N_t=N_{x_t}$. Let $\varsigma_t=Ax_t-b$.
For $a\in\{\beta,-1/p\}$, the coordinatewise closeness estimate in~\cite[Lemma~35]{LS19solving} has coefficient at most
$Cp^2$, since $m^{1/(p+2)}\leq\sqrt e$ and $\abs{a}\leq1/2$. Hence,
\[
    \bnorm{
        \frac{\D}{\D t}
        \log\parenth{w_t^a/\varsigma_t}
    }_\infty
    \leq
    Cp^2\,\norm h_{A_t\tp W_tA_t}\,.
\]

\begin{proof}[Proof of Lemma~\ref{lem:spectral-LS-background-event}]
Write $\varkappa=\varkappa_\delta$ and $T_\delta=c\varkappa^{-4}/\sqrt d$. We use the affine normalization and Gaussian path from \S\ref{sec:spectral-LS-verification}. 
We first construct $\cE_\delta$ and establish the required path bounds.
In addition to the path notation in Fact~\ref{fact:LS-fourth-order-path}, set
\[
    \rho_t\defn W_t^{-1/p}\iota_t\,,\qquad
    z_t\defn\beta\ell_t'-\iota_t\,,\qquad
    D_t=\Diag z_t\,.
\]
Also, Fact~\ref{fact:DWh} and direct differentiation give $\iota_t'=-\iota_t^{\circ2}$ and $\ell_t'=-W_t^{-1/2}N_tv_t$.

Since $g(0)=I_d$, $W_0^\beta A_0$ has orthonormal columns and its $i$-th row has squared norm $w_{0,i}$. Hence,
\[
    \rho_{0,i}
    =w_{0,i}^{-1/2}(W_0^\beta A_0h)_i
    \sim\Normal(0,1).
\]
By Fact~\ref{fact:DWh} and $z_0=\beta\ell_0'-\iota_0$, each coordinate of
$z_0$ is a centered Gaussian linear form. The closeness estimate with
$a=\beta$, together with $A_0\tp W_0A_0\preceq g(0)=I_d$, bounds the
norm of each coefficient vector by $Cp^2$. Gaussian concentration, a union
bound over the rows, and the standard Gaussian norm bound yield an event
$\cE_\delta$ of probability at least $1-\delta$ on which
\begin{align*}
    \norm h&\leq C\varkappa\sqrt d\,,
    &\norm{\rho_0}_\infty+\norm{z_0}_\infty&\leq C\varkappa\,.
\end{align*}
Since $\iota_{0,i}=w_{0,i}^{1/p}\rho_{0,i}$ and $0<w_{0,i}\leq1$,
$\norm{\iota_0}_\infty\leq C\varkappa$, and on the same event,
\[
    \omega_0
    =1+\norm{D_0P_0}_{\op}
    \leq1+\norm{z_0}_\infty
    \leq C\varkappa\,.
\]
The slack identity
$\varsigma_{t,i}=\varsigma_{0,i}(1+t\iota_{0,i})$ and the bound
$T_\delta\norm{\iota_0}_\infty\leq1/2$, after reducing $c$, ensure that the
whole segment $\{x_t:0\leq t\leq T_\delta\}$ lies in $\interior\K$.

Since $\lambda^2g$ is SSC (and hence SC) and $\norm{x_t-x_0}_{\lambda^2g(0)}=t\lambda\norm h\leq Cc\lambda\varkappa^{-3}$,
Lemma~\ref{lem:self-concordant-local-norm-comparison}, together with
$\varkappa\geq p^2$, gives $\norm h_{g(x_t)}\leq C\varkappa\sqrt d$ for $0\leq t\leq T_\delta$.
Since $A_t\tp W_tA_t\preceq g(x_t)$, the RHS of the closeness estimate is at most $C\varkappa^2\sqrt d$ for both $a=\beta$
and $a=-1/p$. Its integral over $[0,T_\delta]$ is at most
$Cc\varkappa^{-2}$. Moreover,
\[
    \rho_{t,i}
    =(a_i\tp h)\,\frac{w_{t,i}^{-1/p}}{\varsigma_{t,i}}\,,\qquad\qquad
    z_t
    =\frac{\D}{\D t} \log\frac{w_t^\beta}{\varsigma_t}\,.
\]
The estimate with $a=-1/p$ gives $\abs{\rho_{t,i}}\leq C\abs{\rho_{0,i}}$, while the estimate with
$a=\beta$ bounds $z_t$ directly. Since $z_t=\beta\ell_t'-\iota_t$ and $\beta$ has a
universal positive lower bound, we obtain, uniformly on the segment,
\begin{align*}
    \norm{\rho_t}_\infty+\norm{\iota_t}_\infty
    &\leq C\varkappa\,,
    &\norm{\ell_t'}_\infty+\norm{z_t}_\infty
    &\leq C\varkappa^2\sqrt d\,.
\end{align*}

Using $(v_t)_i=w_{t,i}^{1/2+1/p}\rho_{t,i}$ and $(q_t)_i=w_{t,i}^{1/2}\rho_{t,i}^2$, together with $\sum_iw_{t,i}=d$ and $0\preceq N_t\preceq pI$, 
\begin{align*}
    \norm{v_t}&\leq C\varkappa\sqrt d,
    &\norm{N_tv_t}+\norm{q_t}
    &\leq C\varkappa^2\sqrt d.
\end{align*}
Direct differentiation, as in the proof of Lemma~3.3 in~\cite{Kook26Dikin}, gives
\begin{align*}
    v_t'
    &=-\iota_t\circ v_t-\frac12\,(N_tv_t)\circ \iota_t\,,\\
    q_t'
    &=\bpar{\frac12-\frac2p}\,\ell_t'\circ q_t
      -2\iota_t\circ q_t\,,
    &\ell_t'\circ q_t&=-(N_tv_t)\circ\rho_t^{\circ2}\,.
\end{align*}
Consequently, $\norm{v_t'} \leq C\varkappa^3\sqrt d$ and $\norm{q_t'} \leq C\varkappa^4\sqrt d$.
Since the maps $t\mapsto D_t$ and $t\mapsto P_{x_t}$ are smooth, the function 
$t\mapsto\omega_t=1+\norm{D_tP_{x_t}}_{\op}$ is continuous on this
segment. These estimates prove all claims of the lemma.
\end{proof}

\end{document}